\documentclass[11pt]{article}
\usepackage[utf8]{inputenc}
\usepackage[top=1in,left=1in,right=1in,bottom=1in]{geometry}
\usepackage{amsthm}
\usepackage{amsmath,mathtools}
\usepackage{subcaption}
\usepackage{color}
\usepackage{xspace}
\definecolor{cite}{rgb}{0.29,0,0.51} %
\definecolor{link}{rgb}{0.5,0,0} %
\definecolor{url}{rgb}{0,0.4,0.4} %
\usepackage[colorlinks=true,citecolor=cite,linkcolor=cite,urlcolor=cite]{hyperref}
\usepackage[capitalise,noabbrev,nameinlink]{cleveref}
\usepackage{CJKutf8} %
\usepackage{float}

\usepackage{xparse}
\ExplSyntaxOn
\cs_new_eq:NN \calc \fp_eval:n
\ExplSyntaxOff

\theoremstyle{plain}
\newtheorem{theorem}{Theorem}[section]
\newtheorem{proposition}[theorem]{Proposition}

\newtheorem{lemma}[theorem]{Lemma}
\newtheorem{corollary}[theorem]{Corollary}

\newcommand{\defn}[1]{\textbf{\emph{\boldmath #1}}}

\newcommand{\plusclue}{\CLUE+ clue\xspace}
\newcommand{\vertclue}{\CLUE| clue\xspace}
\newcommand{\horzclue}{\CLUE- clue\xspace}
\newcommand{\plusclues}{\CLUE+ clues\xspace}
\newcommand{\vertclues}{\CLUE| clues\xspace}

{\makeatletter
 \gdef\xxxmark{%
   \expandafter\ifx\csname @mpargs\endcsname\relax %
     \expandafter\ifx\csname @captype\endcsname\relax %
       \marginpar{xxx}%
     \else
       xxx %
     \fi
   \else
     xxx %
   \fi}
 \gdef\xxx{\@ifnextchar[\xxx@lab\xxx@nolab}
 \long\gdef\xxx@lab[#1]#2{\textbf{[\xxxmark #2 ---{\sc #1}]}}
 \long\gdef\xxx@nolab#1{\textbf{[\xxxmark #1]}}
 \long\gdef\xxx@lab[#1]#2{}\long\gdef\xxx@nolab#1{}%
}

\def\andlinebreak{\end{tabular}\linebreak\begin{tabular}[t]{c}}

\usepackage{microtype}

\def\CLUE#1{\edef\CLUEarg{#1}%
  \hspace{.1em}%
  \if+\CLUEarg
    \includegraphics[height=1.75ex]{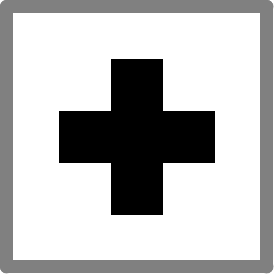}%
  \else
    \if-\CLUEarg
      \includegraphics[height=1.75ex]{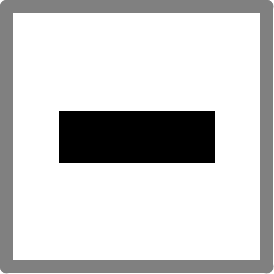}%
    \else
      \if|\CLUEarg
        \includegraphics[height=1.75ex]{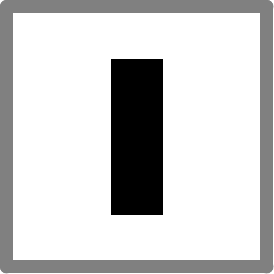}%
      \fi
    \fi
  \fi
  \hspace{.1em}}

\newlength\WIDTH
\WIDTH=\linewidth
\edef\SCALE{\calc{\WIDTH / 2178.75 * 72 / 72.27}} %

\title{Tatamibari is NP-complete}

\author{%
  Aviv Adler%
    \thanks{MIT Computer Science and Artificial Intelligence Laboratory,
      32 Vassar St., Cambridge, MA 02139, USA.
      \protect\url{{adlera,jbosboom,edemaine,quanquan,jaysonl}@mit.edu}}
\and
  Jeffrey Bosboom\footnotemark[1]
\and
  Erik D. Demaine\footnotemark[1]
\andlinebreak
  Martin L. Demaine\footnotemark[1]
\and
  Quanquan C. Liu\footnotemark[1]
\and
  Jayson Lynch\footnotemark[1]
}

\date{}

\begin{document}

\maketitle

\begin{abstract}
  In the Nikoli pencil-and-paper game Tatamibari, a puzzle consists of an
  $m \times n$ grid of cells, where each cell possibly contains a clue among
  \CLUE+, \CLUE-, \CLUE|.  The goal is to partition the grid into disjoint
  rectangles, where every rectangle contains exactly one clue,
  rectangles containing \CLUE+ are square, rectangles containing \CLUE-
  are strictly longer horizontally than vertically, rectangles containing
  \CLUE| are strictly longer vertically than horizontally, and
  no four rectangles share a corner.
  We prove this puzzle NP-complete, establishing a Nikoli gap of 16 years.
  Along the way, we introduce a gadget framework for proving hardness of
  similar puzzles involving area coverage, and show that it applies to an
  existing NP-hardness proof for Spiral Galaxies.
  We also present a mathematical puzzle font for Tatamibari.
\end{abstract}

\section{Introduction}

Nikoli is perhaps the world leading publisher of pencil-and-paper logic puzzles,
having invented and/or popularized hundreds of different puzzles through their
\emph{Puzzle Communication Nikoli} magazine and hundreds of books.
Their English website \cite{nikoli-puzzles} currently lists 38 puzzle types,
while their ``omopa list'' \cite{omopalist} currently lists 456 puzzle types
and their corresponding first appearance in the magazine.

Nikoli's puzzles have drawn extensive interest by theoretical computer
scientists (including the FUN community): whenever a new puzzle type gets
released, researchers tackle its computational complexity.
For example, the following puzzles are all NP-complete:
Bag / Corral \cite{Friedman-2002-corral},
Country Road \cite{YajilinandCountryRoad},
Fillomino \cite{Yato-2003},
Hashiwokakero \cite{Hashiwokakero},
Heyawake \cite{Holzer-Ruepp-2007},
Hiroimono / Goishi Hiroi \cite{Andersson-2007},
Hitori \cite[Section~9.2]{GPC},
Kakuro / Cross Sum \cite{Yato-Seta-2003},
Kurodoko \cite{Kurodoko},
Light Up / Akari \cite{McPhail-2005-LightUp},
LITS \cite{McPhail-2007-LITS},
Masyu / Pearl \cite{Friedman-2002-pearl},
Nonogram / Paint By Numbers \cite{Ueda-Nagao-1996},
Numberlink \cite{NumberlinkNP,ADDORVS15},
Nurikabe \cite{McPhail-2003,Holzer-Klein-Kutrib-2004},
Shakashaka \cite{Shakashaka,NumberlessShakashaka_CCCG2015},
Slitherlink \cite{Yato-Seta-2003,Yato-2003,Witness_FUN2018},
Spiral Galaxies / Tentai Show \cite{SpiralGalaxies},
Sudoku \cite{Yato-Seta-2003,Yato-2003},
Yajilin \cite{YajilinandCountryRoad},
and
Yosenabe \cite{Yosenabe}.

Allen et al.~\cite{Nikoli-gap} defined the \defn{Nikoli gap} to be the amount
of time between the first publication of a Nikoli puzzle and a hardness result
for that puzzle type.
They observed that, while early Nikoli puzzles have a gap of 10--20 years,
puzzles released within the past ten years tend to have a gap of $<5$ years.

In this paper, we prove NP-completeness of a Nikoli puzzle first published
in 2004 \cite{tatamibari-original} (according to \cite{omopalist}),
establishing a Nikoli gap of 16 years.%
\footnote{While this gap may be caused by the puzzle being difficult to prove
  hard or simply overlooked (or both), we can confirm that it took us nearly
  six years to write this paper.}
Specifically, we prove NP-completeness of the Nikoli puzzle \defn{Tatamibari}
(\begin{CJK}{UTF8}{min}タタミバリ\end{CJK}), named after Japanese tatami mats.
A Tatamibari \defn{puzzle} consists of a rectangular $m \times n$ grid of
unit-square cells, some $k$ of which contain one of three different kinds of
clues: \CLUE+, \CLUE|, and \CLUE-.
(The remaining $m \cdot n - k$ cells are empty, i.e., contain no clue.)
A \defn{solution} to such a puzzle is a set of $k$ grid-aligned rectangles
satisfying the following constraints:
\begin{enumerate}
\item The rectangles are disjoint.
\item The rectangles together cover all cells of the puzzle.
\item Each rectangle contains exactly one symbol in it.
\item The rectangle containing a \CLUE+ (``square'') symbol is a square, i.e., has equal width (horizontal dimension) and height (vertical dimension).
\item The rectangle containing a \CLUE- (``horizontal'') symbol has greater width than height.
\item The rectangle containing a \CLUE| (``vertical'') symbol has greater height than width.
\item No four rectangles share the same corner (\defn{four-corner constraint}).
\end{enumerate}

To prove our hardness result, we first introduce in \cref{sec:framework}
a general ``gadget area hardness framework'' for arguing about (assemblies of)
local gadgets whose logical behavior is characterized by area coverage.
Then we apply this framework to prove Tatamibari NP-hard in
\cref{sec:main-section}.
In \cref{sec:spiral-galaxies}, we show that our framework applies to at least
one existing NP-hardness proof, for the Nikoli puzzle Spiral Galaxies
\cite{SpiralGalaxies}.

We also present in \cref{sec:font} a mathematical puzzle font \cite{Fonts_TCS}
for Tatamibari, consisting of 26 Tatamibari puzzles whose solutions draw each
letter of the alphabet.  This font enables writing secret messages, such as
the one in \cref{fig:FUN}, that can be decoded by solving the Tatamibari
puzzles.  This font complements a similar font for another Nikoli puzzle,
Spiral Galaxies \cite{SpiralGalaxies_MOVES2017}.

\def\FUNscale{1.4}

\begin{figure}
  \centering
  \includegraphics[scale=\FUNscale]{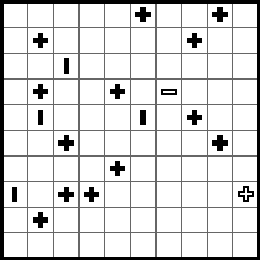}\qquad
  \includegraphics[scale=\FUNscale]{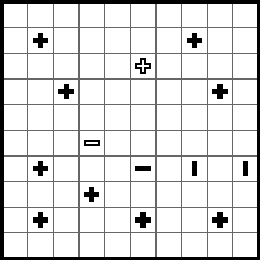}\qquad
  \includegraphics[scale=\FUNscale]{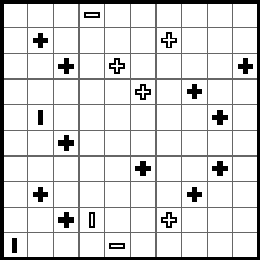}
  \caption{What do these Tatamibari puzzles spell when solved and the dark
    clues' rectangles are filled in?  \cref{fig:FUN solved} gives a solution.}
  \label{fig:FUN}
\end{figure}

\section{Gadget Area Hardness Framework} \label{sec:framework}

\paragraph{Puzzles.}
The \defn{gadget area hardness framework} applies to a general
\defn{puzzle type} (e.g., Tatamibari or Spiral Galaxies)
that defines puzzle-specific mechanics.
In general, a \defn{subpuzzle} is defined by an embedded
planar graph, whose finite faces are called \defn{cells}, together with
an optional \defn{clue} (e.g., number or symbol) in each cell.
The puzzle type defines which subpuzzles are valid \defn{puzzles},
in particular, which clue types and planar graphs are permitted,
as well as any additional \textbf{properties} guaranteed by a hardness reduction
producing the puzzles.

We will use the unrestricted notion of subpuzzles to define gadgets.
Define an \defn{area} of a puzzle to be a connected set of cells.
An \defn{instance} of a subpuzzle in a puzzle is an area of the puzzle
such that the restriction of the puzzle to that area (discarding all cells
and clues outside the area) is exactly the subpuzzle.

\paragraph{Solutions.}
An \defn{area assignment} (potential solution) for a (sub)puzzle
is a mapping from clues to areas such that
(1)~the areas disjointly partition the cells of the (sub)puzzle,
and (2)~each area contains the cell with the corresponding clue.
The puzzle type defines when an area assignment is an actual \defn{solution}
to a puzzle or a \defn{local solution} to a subpuzzle.

\xxx{Running example of Tatamibari?}

\paragraph{Gadgets.}
A \defn{gadget} is a subpuzzle plus a partition of its \defn{entire} area
(all of its cells)
into one \defn{mandatory} area and two or more \defn{optional} areas,
where all clues are in the mandatory area.
A hardness reduction using this framework should compose puzzles from instances
of gadgets that overlap only in optional areas, and provide a
\defn{filling algorithm} that defines which clues are in the areas exterior
to all gadgets.
Each gadget thus defines the entire set of clues of the puzzle
within the gadget's (mandatory) area.

For a given gadget,
a \defn{gadget area assignment} is an area assignment for the subpuzzle
that satisfies three additional properties:
\begin{enumerate}
\item the assigned areas cover the gadget's mandatory area;
\item every optional area is either fully covered or fully uncovered
      by assigned areas; and
\item every assigned area lies within the gadget's entire area.
\end{enumerate}
A \defn{gadget solution} is a gadget area assignment that is a local solution
as defined by the puzzle type.

\paragraph{Profiles.}
A \defn{profile} of a gadget is a subset of the gadget's entire area.
A profile is \defn{proper} if it satisfies two additional properties:
\begin{enumerate}
\item the profile contains the mandatory area of the gadget; and
\item every optional area of the gadget is either fully contained or
      disjoint from the profile.
\end{enumerate}
Every gadget area assignment induces a proper profile, namely,
the union of the assigned areas.

A profile is \defn{locally solvable} if there is a gadget solution
with that profile.
A profile is \defn{locally impossible} if, in any puzzle containing an
instance of the gadget, there is no solution to the entire puzzle
such that the union of the areas assigned to the clues of the gadget instance
is that profile.
These notions might not be negations of each other because of differences
between local solutions of a subpuzzle and solutions of a puzzle.

Each gadget is characterized by a \defn{profile table} (like a truth table)
that lists all profiles that are locally solvable, and for each such profile,
gives a gadget solution.
A profile table is \defn{proper} if it contains only proper profiles.
A profile table is \defn{complete} if every profile not in the table
is locally impossible.
A hardness reduction using this framework should prove that
the profile table of each gadget is proper and complete,
in particular, that any improper profile is locally impossible.

Given a puzzle containing some gadget instances,
a \defn{profile assignment} specifies a profile for each gadget such that
the profiles are pairwise disjoint and the union of the profiles covers the
union of the entire areas of the gadgets.
In particular, such an assignment decides which overlapping optional areas
are covered by which gadgets.
A profile assignment is \defn{valid} if every gadget is locally solvable
with its assigned profile, i.e., every assigned profile is in the profile table
of the corresponding gadget.

A hardness reduction using this framework should prove that
every valid profile assignment can be extended to a solution
of the entire puzzle by giving a \defn{composition algorithm}
for composing local solutions from the profile tables of the gadgets,
possibly modifying these local solutions to be globally consistent,
and extending these solutions to assign areas to clues
from the filling algorithm (exterior to all gadgets).

\xxx{contrast from ``local puzzle constraints'' with finite set of stamps around each clue?}

\section{Tatamibari is NP-hard} \label{sec:main-section}

In this section, we prove Tatamibari NP-hard by a reduction from
planar rectilinear monotone 3SAT.
In \cref{sec:prm-3sat}
 we briefly discuss a more constrained (but still NP-hard) variant of the classic 3SAT problem from which we will make our reduction; in \cref{sec:wires}, \cref{sec:variables}, and \cref{sec:clauses}, we describe the gadgets (wires, variables, and clauses) from which we build the reduction; in \cref{sec:filler}, we discuss how the spaces between the gadgets are filled; and in \cref{sec:finale} we use everything to show the main result.

\subsection{Reduction Overview} \label{sec:prm-3sat}

We reduce from \emph{planar rectilinear monotone 3SAT}, proved NP-hard in~\cite{DK10}.  An instance of planar rectilinear monotone 3SAT comes with a planar rectilinear drawing of the clause-variable graph in which each variable is a horizontal segment on the $x$-axis and each clause is a horizontal segment above or below the axis, with rectilinear edges connecting variables to the clauses in which they appear.  Each clause contains only positive or negative literals (i.e., is monotone); clauses containing positive (negative) literals appear above (below) the variables.  We can always lengthen the variable and clause segments to remove bends in the edges, so we assume the edges are vertical line segments.
We can further assume that each clause consists of exactly three (not necessarily distinct) literals: if a clause has $k < 3$ literals, we can just duplicate one of the clauses $3-k$ times, which is easy to do while preserving the tri-legged rectilinear layout.

We create and arrange our gadgets directly following the drawing, possibly after scaling it up; see \cref{fig:filler-diagram}.  Edges between variables and clauses are represented by \emph{wire gadgets} that communicate a truth value in the parity of their covering.  For each variable, we create a \emph{variable gadget}, which is essentially a block of wires forced to have the same value, and place it to fill the variable's line segment in the drawing.  For each clause, we create a \emph{clause gadget} with three wire connection points and place it to fill the clause's line segment.  Negative clauses and wires representing negative literals are mirrored vertically.  Both the variable and clause gadgets can telescope to any width to match the drawing; unused wires from the variable gadgets are terminated at a \emph{terminator}.  By our assumption that the edges are vertical segments, we do not need a turn gadget.

\begin{figure}
\centering
\includegraphics[width=0.9\linewidth]{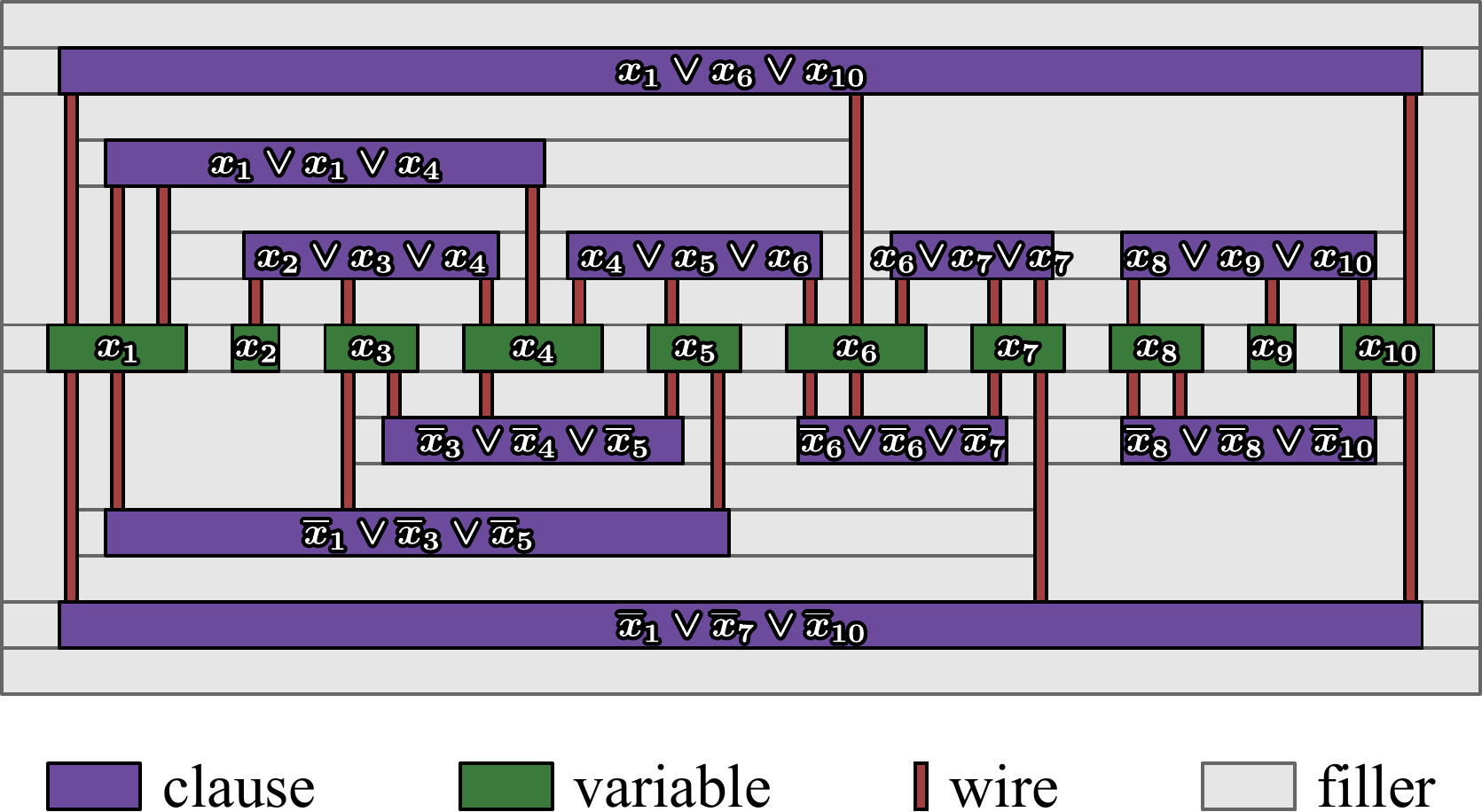}
\caption{The overall layout of the Tatamibari puzzles produced by our reduction follows the input planar rectilinear monotone 3SAT instance.
Clause, variable, and wire gadgets are represented by
purple, green, and red rectangles.
Not drawn are terminator gadgets at the base of all unused copies of variables.
Grey rectangles correspond to individual filler clues.
Figure inspired by \cite[Figure 2]{DK10}.}
\label{fig:filler-diagram}
\end{figure}

Covering a clause gadget without double-covering or committing a four-corner violation requires at least one of its attached wires to be covered with the satisfying parity (the true parity for positive clauses and the false parity for negative clauses).

To ensure clues in one gadget do not interfere with other gadgets, the wire gadget is surrounded on its left and right sides by sheathing of \vertclue rectangles and the clause gadget is surrounded on three sides by a line of \CLUE+ clues forced to form $1 \times 1$ rectangles.  Wire sheathing also ensures neighboring wires do not constrain each other, except in variable gadgets where the sheathing is deliberately punctured.

In our construction, gadgets will not overlap in their mandatory areas,
so in the intended solutions, the mandatory area will be fully covered by
rectangles satisfying the gadget's clues.
Also in our construction, every optional area will belong to exactly two
gadgets, and in the intended solutions, such an area
will be covered by clues in exactly one of those gadgets.

To apply the gadget area hardness framework, we define
a \defn{local solution} of a subpuzzle to be a disjoint set of rectangles
satisfying the gadget's clues and the property that
no four of these rectangles share a corner.
(At the boundary of the subpuzzle, there is no constraint.)
Our composition algorithm will combine these local solutions by
staggering rectangles to avoid four-corner violations on the boundary of
and exterior to gadgets.
We will prove that valid profile assignments correspond one-to-one to
satisfying truth assignments of the 3SAT instance.

We developed our gadgets using a Tatamibari solver based on the SMT solver Z3~\cite{DBLP:conf/tacas/MouraB08}.  The solver and machine-readable gadget diagrams are available \cite{github}\xxx{put them on GitHub (in a sanitized repo)}.
Unfortunately, the solver can only verify the correctness of
constant-size instances of the gadgets, but the variable and clause gadgets must telescope to arbitrary width.  Thus we still need to give manual
proofs of correctness.

\subsection{Wire Gadgets and Terminators} \label{sec:wires}
The wire gadget, shown in \cref{fig:wire}, consists of a column of \plusclues surrounded by \vertclues which encodes a truth value in the parity of whether the squares are oriented with the \plusclues in their upper left or lower left corners. We will call this the \emph{wire parity} or \emph{wire value}. In this construction, only vertical wires are needed, and thus we do not give a turn gadget or horizontal wire. We call the column containing the \plusclues and the empty column next to it the \emph{inner wire}.  The inner wire is covered by columns of alternating \vertclues, called the \emph{(inner) sheathing}.  In the overall reduction, \vertclues in columns just outside the wire at its ends (in the variable gadget and either the clause or terminator gadget) add a further layer of sheathing (called the \emph{outer sheathing}) outside the wire gadget, ensuring neighboring wires do not constrain each other.

The following lemmas will show that the \plusclues in the wire must be covered by $2\times 2$ squares, the squares must all have the same parity, and the wire does not impart any significant constraints onto the surrounding region.  These lemmas assume that no rectangle from a \vertclue can reach the cells to the right of the top and bottom \plusclues in the wire, a property which we call \defn{safe placement}.  We discharge this assumption in Section~\ref{sec:filler} by showing all wire gadgets produced by our reduction are safely placed.

\begin{figure}
\centering
\begin{subfigure}[t]{0.45\linewidth}
  \centering
  \includegraphics[scale=\SCALE]{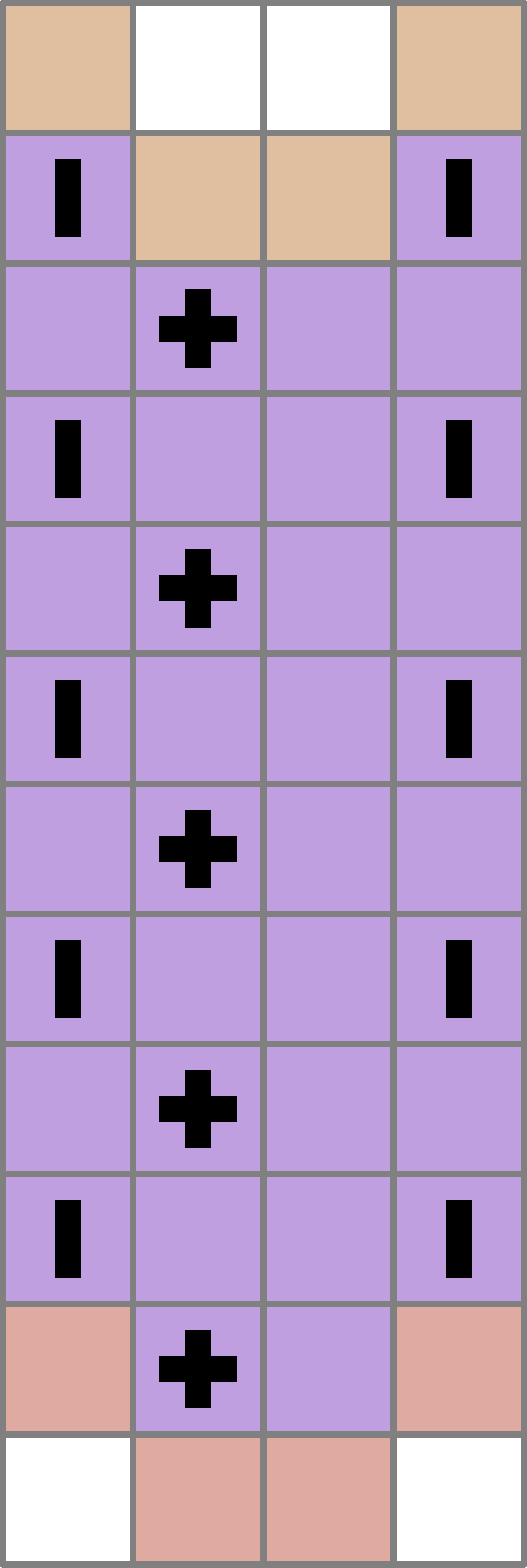}
  \caption{An unsolved wire gadget. Mandatory area is purple and optional areas are brown.}
  \label{fig:wire-unsolved}
\end{subfigure}\hfill
\begin{subfigure}[t]{0.25\linewidth}
  \centering
  \includegraphics[scale=\SCALE,trim=0 -75 0 0]{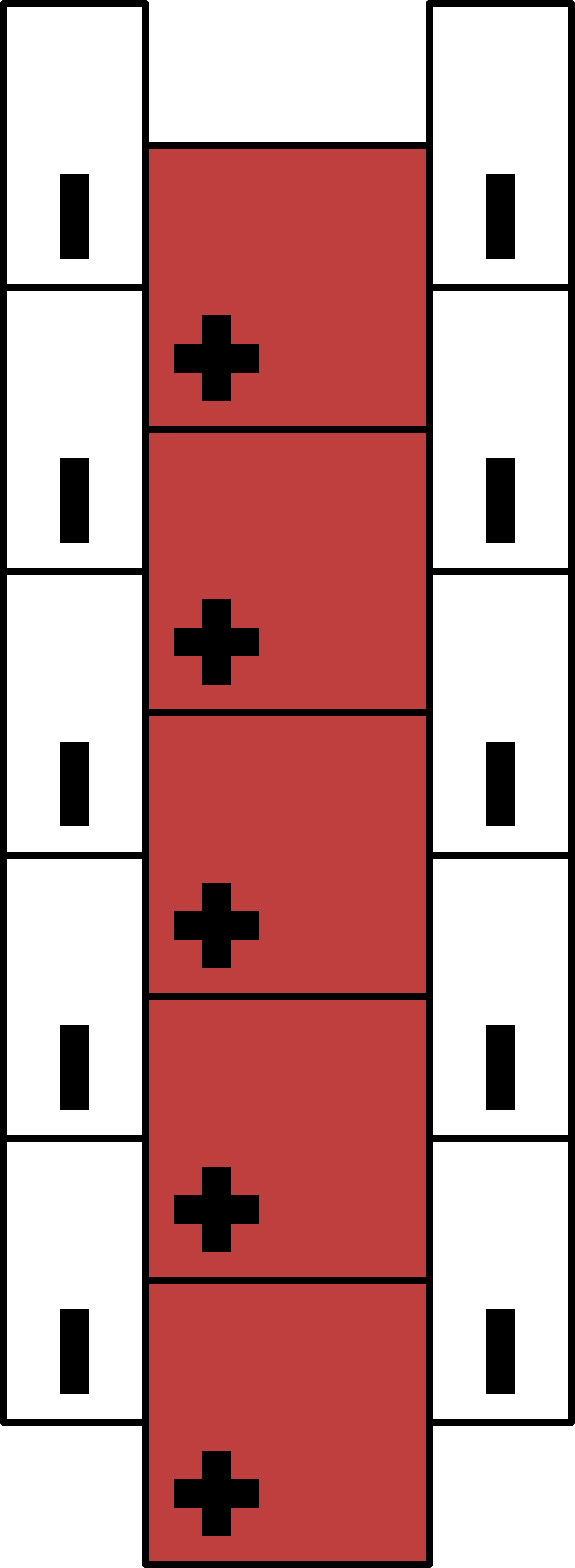}
  \caption{Wire communicating false}
  \label{fig:wire-0}
\end{subfigure}\hfill
\begin{subfigure}[t]{0.25\linewidth}
  \centering
  \includegraphics[scale=\SCALE]{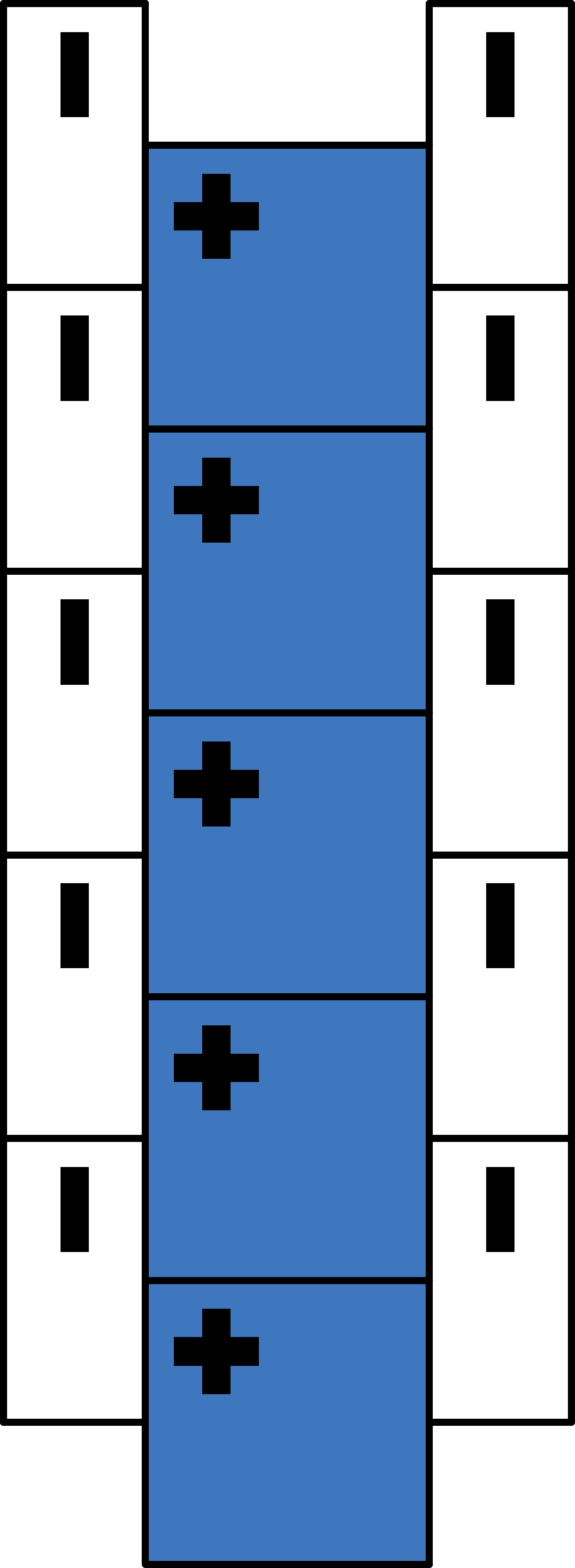}
  \caption{Wire communicating true}
  \label{fig:wire-1}
\end{subfigure}
\caption{Wire gadget and its profile table.  The wire can be extended to arbitrary height by repeating rows. Note that between figures (b) and (c), the clues stay in the same place (and the rectangles shift to represent the different values of the wire).}
\label{fig:wire}
\end{figure}

\begin{figure}
\centering
\begin{subfigure}[t]{0.45\linewidth}
  \centering
  \includegraphics[scale=\SCALE]{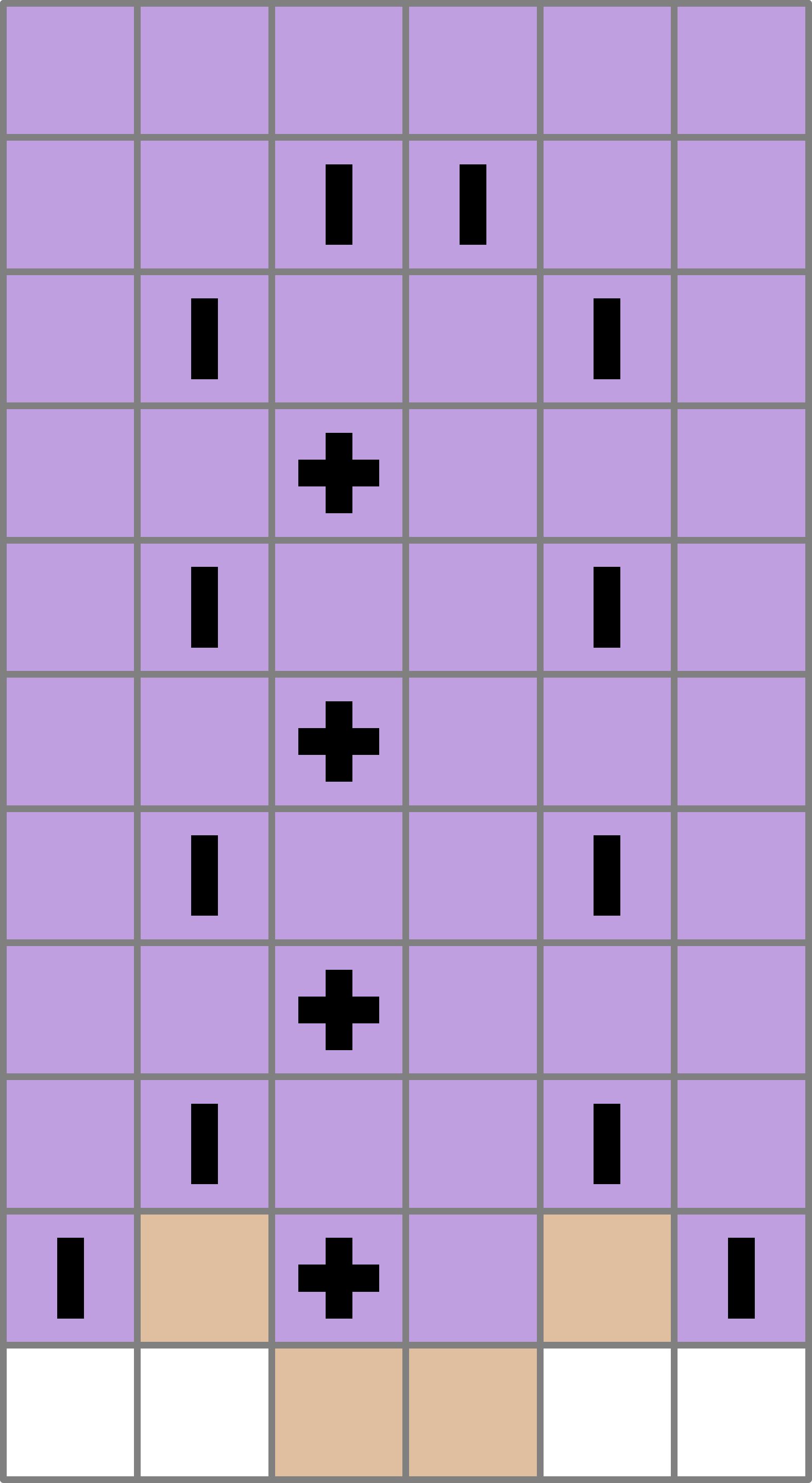}
  \caption{An unsolved terminator gadget. Mandatory area is purple and optional area is brown.}
\end{subfigure}\hfill
\begin{subfigure}[t]{0.25\linewidth}
  \centering
  \includegraphics[scale=\SCALE,trim=0 -75 0 0]{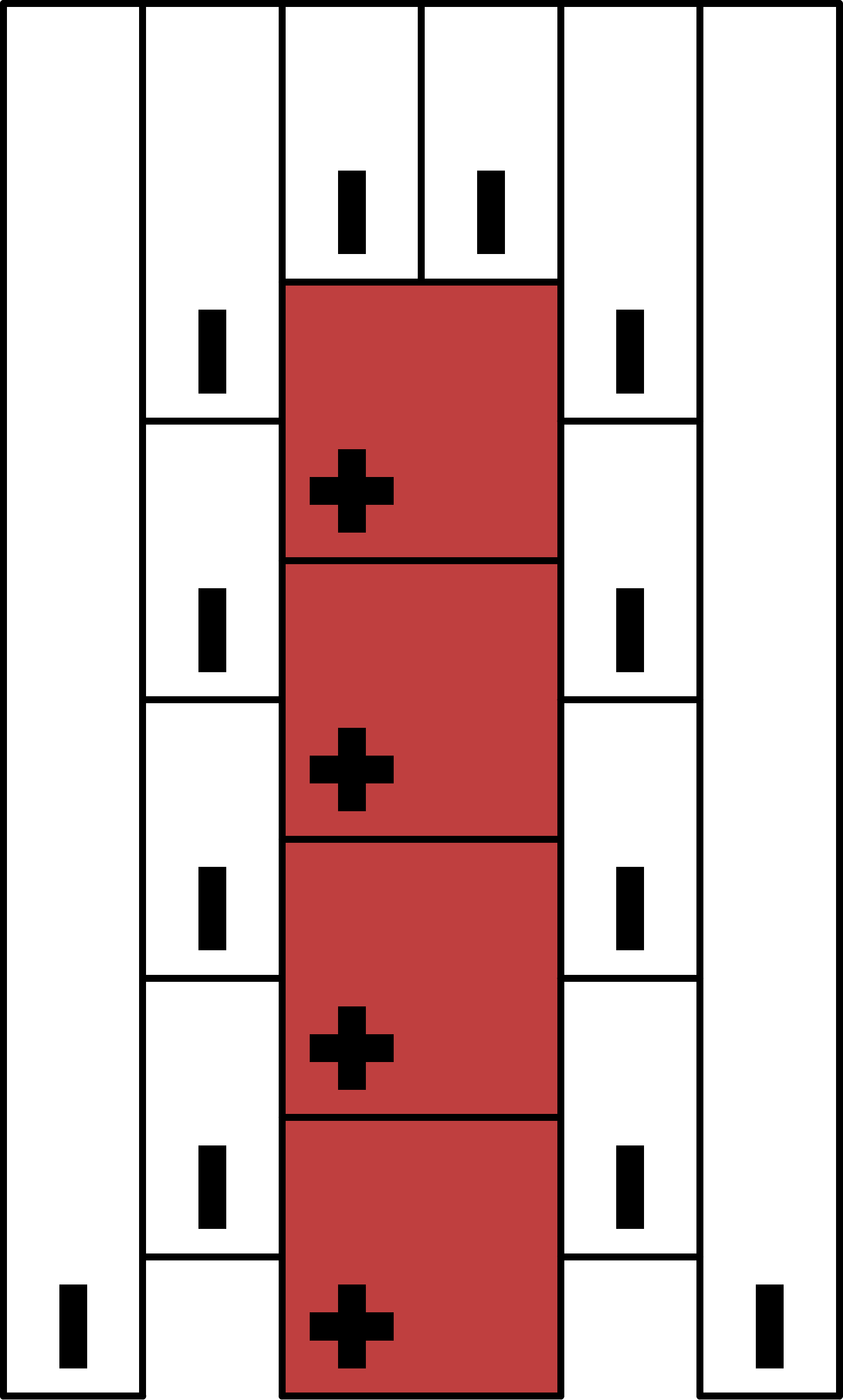}
  \caption{Terminating a false wire}
  \label{fig:terminate-0}
\end{subfigure}\hfill
\begin{subfigure}[t]{0.25\linewidth}
  \centering
  \includegraphics[scale=\SCALE]{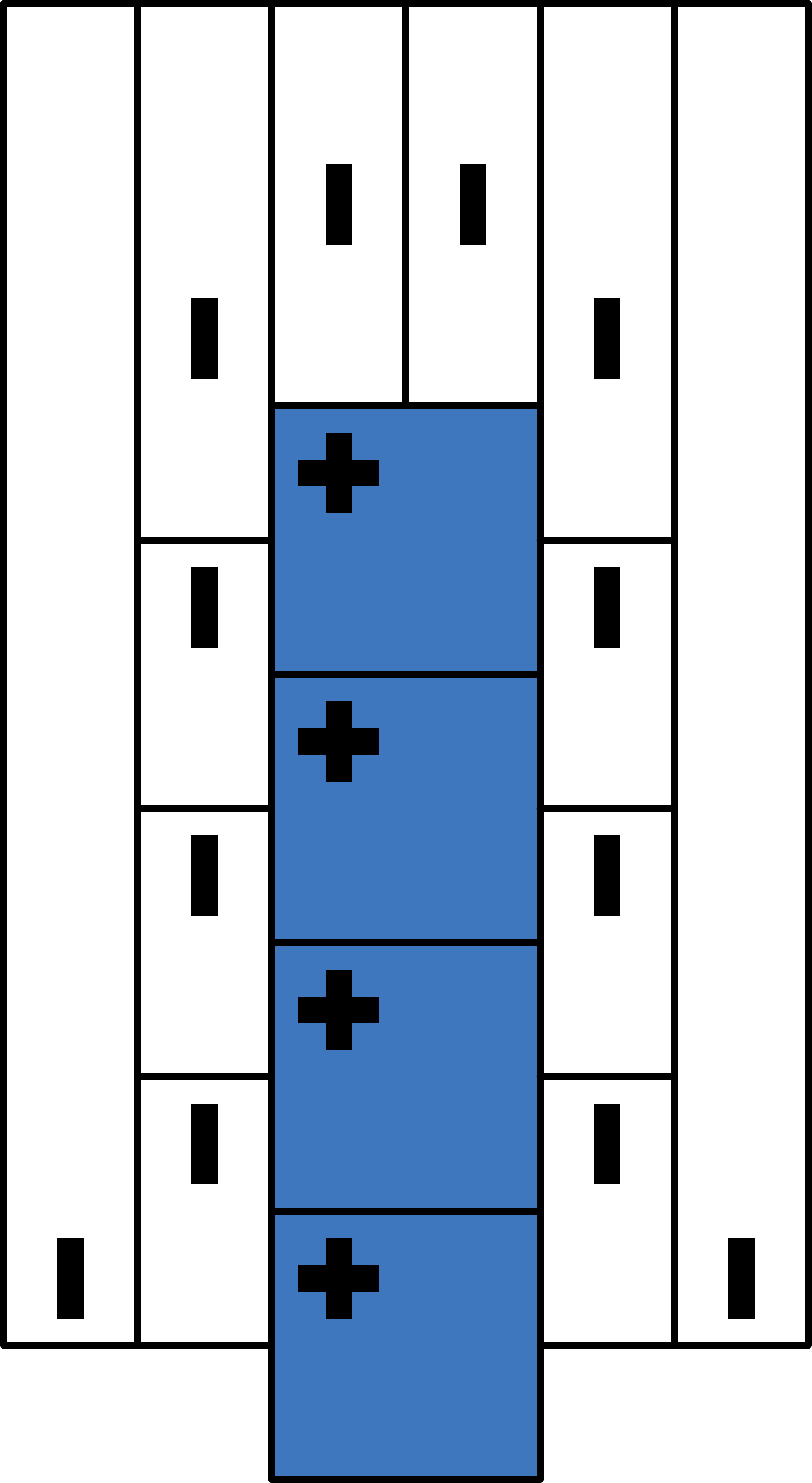}
  \caption{Terminating a true wire}
  \label{fig:terminate-1}
\end{subfigure}
\caption{Terminator gadget and its profile table.}
\label{fig:terminator}
\end{figure}

\begin{lemma} \label{lem:wire2x2}
  Each \CLUE+ in a safely placed wire covers a $2 \times 2$ square in the wire.
\end{lemma}

\begin{proof}
  There is no $3\times 3$ square in the wire that contains a \plusclue but does not contain any other clue. Thus we cannot cover the \plusclue by squares larger than $2\times 2$.

  Now suppose we cover a \plusclue by a $1\times 1$ square. Now the cells immediately above and below this clue must be covered. The \vertclues must be taller than they are long, so they cannot cover these cells. Thus we must cover them by squares containing the \plusclues above and below. This leaves the cell directly to the right of the $1\times 1$ uncovered. It is easy to see that this cannot be covered by the nearby \plusclues or \vertclues.  The cells next to the top and bottom clues cannot be covered by a \vertclue from outside the gadget by our assumption that the wire is safely placed.

The only remaining possibility is a \horzclue from outside the gadget extending into the wire gadget.  Such a rectangle cannot extend entirely through the wire because the \vertclues in the sheathing and the \plusclues inside the wire are in alternating rows.  If the external horizontal rectangle enters the wire from the right and covers a cell next to the \plusclue, that \plusclue is forced to be a $1 \times 1$ rectangle and the cell above it must be covered by the next \plusclue above.  This results in a four-corner violation involving the two \plusclues and the left sheathing except when the \plusclue is at the bottom of the wire.  In that case, the external horizontal rectangle blocks the bottom-right sheathing clue, making it $1 \times 1$ and unsatisfied.
\end{proof}

\begin{corollary}\label{cor:wire-parity}
	Satisfied safely placed wires must have all of their $2\times 2$ squares with the \plusclues in the lower left corner or all in the upper left corner.
\end{corollary}
\begin{proof}
	By \cref{lem:wire2x2} all \plusclues must be covered by $2\times 2$ squares. To change whether the \plusclues are in the lower left or upper left, we will end up leaving a row of two cells between clues blank. By the same arguments in \cref{lem:wire2x2} these cannot be covered by the nearby \plusclues or \vertclues.
\end{proof}

\begin{lemma}\label{thm:sheathing-parity}
  The \vertclues making up the inner sheathing of satisfied safely placed wires must be covered by $1 \times 2$ rectangles of opposite parity to the wire's squares.
\end{lemma}
\begin{proof}
	By \cref{cor:wire-parity} the wire has one of two parities of squares. If a vertical rectangle ends with the same $y$ coordinate as an adjacent square, then we will have three right angles at a single corner, forcing a four-corner violation or uncovered cell. Because the squares are $2\times 2$, a vertical rectangle of odd height guarantees one of the ends will share a $y$-coordinate with one of the squares. The \vertclues occur every other cell, so the vertical rectangles cannot be of length greater than $3$. This forces them to be of length $2$ and staggered with respect to the squares.
\end{proof}

\begin{theorem}
The safely placed wire gadget's profile table is proper and complete.
\end{theorem}
\begin{proof}
By \cref{lem:wire2x2}, each optional area must be fully covered or fully uncovered by the neighboring \plusclue, so the profile table is proper.  \cref{cor:wire-parity} fixes the \plusclue parity and \cref{thm:sheathing-parity} fixes the sheathing parity, so all other profiles are locally impossible, so the profile table is complete.
\end{proof}

We also have a terminator gadget to terminate unused wires regardless of their parity.  The terminator gadget is shown in \cref{fig:terminator}.

\begin{lemma} \label{terminal-wire independence}
  The terminator does not constrain the wire parity.
\end{lemma}
\begin{proof}
    \cref{fig:terminate-0,fig:terminate-1} show solutions of the terminator with both parities.  The same arguments about wire correctness show this gadget does not allow any additional wire solutions nor constrain other gadgets.
\end{proof}

\begin{theorem}
The terminator gadget's profile table is proper and complete.
\end{theorem}
\begin{proof}
The profile table in \cref{fig:terminator} contains only proper profiles, so the profile table is proper.  By the same arguments in \cref{lem:wire2x2}, the two local solutions shown are the only way to cover the wire part of the gadget.  A horizontal rectangle from a \CLUE- outside the gadget could cover part of the top row of the gadget (or the entire top row when terminating a true wire) while leaving the clues in the gadget satisfied and covering the remaining area.  We prevent this through the global layout: all clause gadgets (the only gadget containing \CLUE- clues) appear strictly above (for positive clauses) or strictly below (for negative clauses) all terminator gadgets, so it is not possible for any \CLUE- rectangles to cover area in the clause gadget.  Thus all other profiles are locally impossible, so the profile table is complete.
\end{proof}

\subsection{Variable Gadgets} \label{sec:variables}

The variable gadget is essentially a series of wires placed next to each other with devices we call \defn{couplers} in between. Each coupler acts essentially as an ``equality'' constraint between neighboring wires, thus forcing all the wires connected via a series of couplers to represent the same variable; this collection of wires then forms the \defn{variable gadget} of the reduction.

Each coupler takes two columns, and consists of (i) a \plusclue which interacts with the inner sheathing of the wires to force equality, and (ii) eight \vertclues (two above and two below the \plusclue on each column), which prevent the inner sheathings of the neighboring wires from constraining each other (except through the \plusclue itself). See \cref{fig:variable} for an example with two wires; additional wires can be added to either side of variable by using more couplers (see \cref{fig:variable-wide}).

First, notice that both wires are constrained to have their squares in one of two parities by the inner sheathing, as in \cref{cor:wire-parity}. It is also important that wires do not constrain each other outside the couplers, either directly (if they happen to be adjacent) or indirectly (through the space in between); we address this in \cref{sec:filler}.

Now we have to show that two wires separated by the coupler must be in the same configuration. This happens because the wire parity forces the parity of the inner sheathing, which forces the parity of the coupler, which then forces the partiy of the inner sheathing and the wire parity of the next wire over.

\begin{lemma}\label{lem:coupler-configs}
	The coupler has only two valid coverings of its \plusclue.
\end{lemma}
\begin{proof}
	The location of the eight \vertclues around the \plusclue ensure that it cannot be larger than $2 \times 2$. By \cref{cor:wire-parity} we know that the wire gadgets next to the coupler must have their inner sheathing as $2\times1$ rectangles in either the up or down position. If the \plusclue is covered by a $1 \times 1$ it will create a four-corner violation with the inner sheathing. Thus it must be one of the two possible positions for a $2 \times 2$ square. If both inner sheathings have the same parity, as in \cref{fig:variable} then the constraints can be locally satisfied.
\end{proof}

\begin{lemma} \label{lem: variable gadget}
  All wires in a variable gadget must have the same value (i.e. upwards branches must have the same orientation).
\end{lemma}

\begin{proof}
  We know the coupler has at most two ways to satisfy its constraints, corresponding to a $2 \times 2$ square in either the up or down position. Notice that the inner sheathing of both wires must be of different parity from the square or they will cause a four-corner violation. Thus the inner sheathing must have the same parity, ensuring that the wires themselves must have the same parity. If multiple wires are all connected by couplers, then they will all be forced to have the same parity by the same local argument.
\end{proof}

\begin{lemma} \label{lem: variable valid}
  The variable gadget is locally solvable with a given profile if and only if the profile satisfies (i) all upwards branches have the same orientation, (ii) all downwards branches have the same orientation, and (iii) upwards and downwards branches have opposite orientations from each other.
\end{lemma}

\begin{proof}
The ``only if'' direction follows from \cref{lem: variable gadget}  and  \cref{cor:wire-parity} (each wire individually must have opposite orientation for upwards and downwards branches due to the couplers, and all wires in the gadget must have the same upward orientation).

The ``if'' direction follows from \cref{terminal-wire independence}, the individual solvability of each wire and terminator in both orientations (as shown in \cref{fig:wire-0}, \cref{fig:wire-1}, \cref{fig:terminate-0}, and \cref{fig:terminate-1}), and the solvability of the couplers given that adjacent wires have the same orientation (\cref{fig:variable}). Neighboring wires (within the variable gadget) do not conflict with each other (outside of the coupler) because of the ``outer sheathing'' columns separating them; the meeting points of the two clues in each ``outer sheathing'' column can be adjusted to avoid four-corner violations with each other, as well as avoiding four-corner violations with the neighboring ``inner sheathing''.
\end{proof}
Note that this lemma is what we want from a variable gadget: it is locally solvable if and only if its profile corresponds to a specific value for the variable it represents.

\begin{figure}
\centering
\includegraphics[scale=\SCALE]{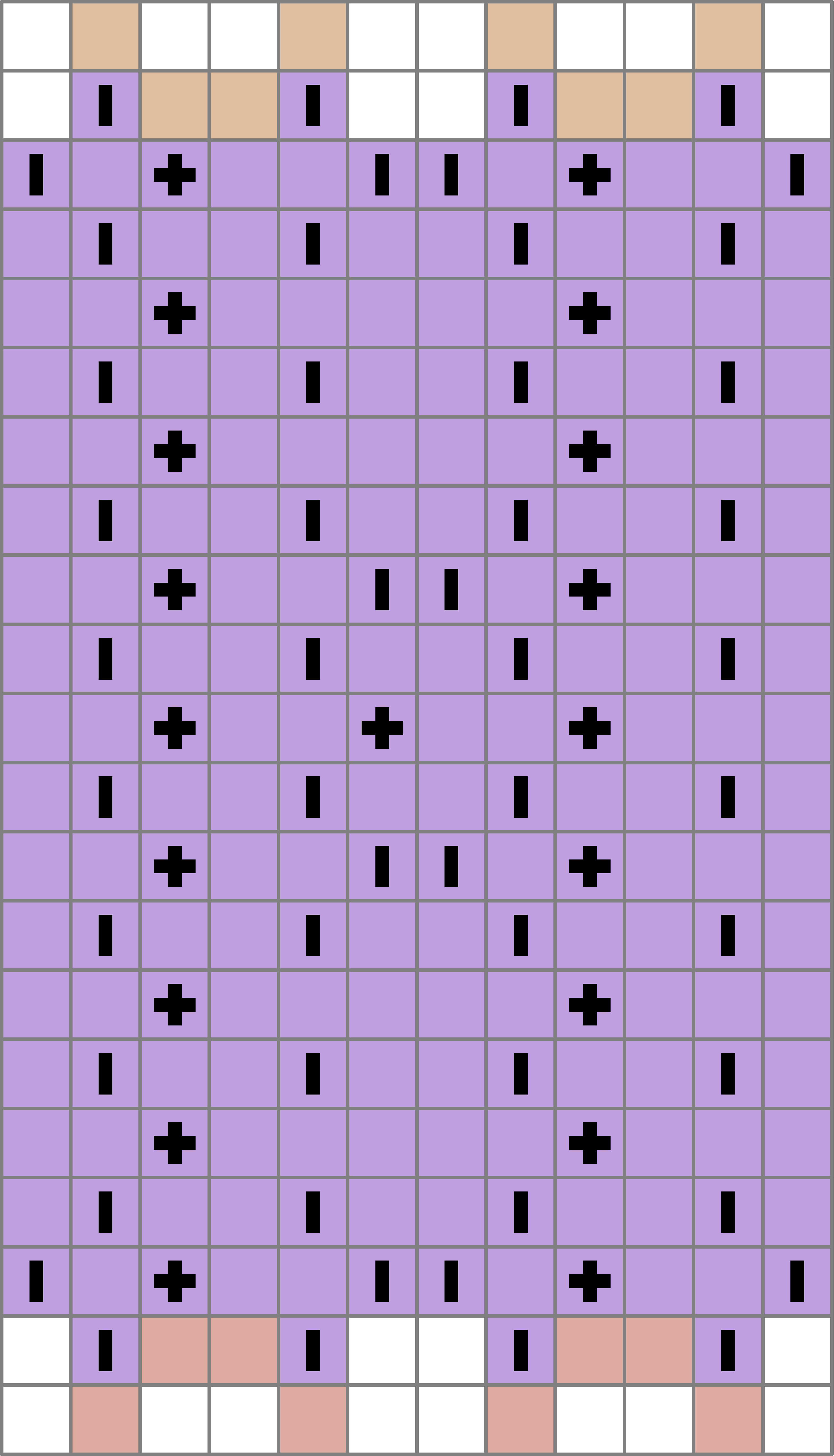}
\caption{The variable gadget.  Mandatory area is purple and optional areas are brown.}
\end{figure}

\begin{figure}%
\centering
\includegraphics[scale=\SCALE]{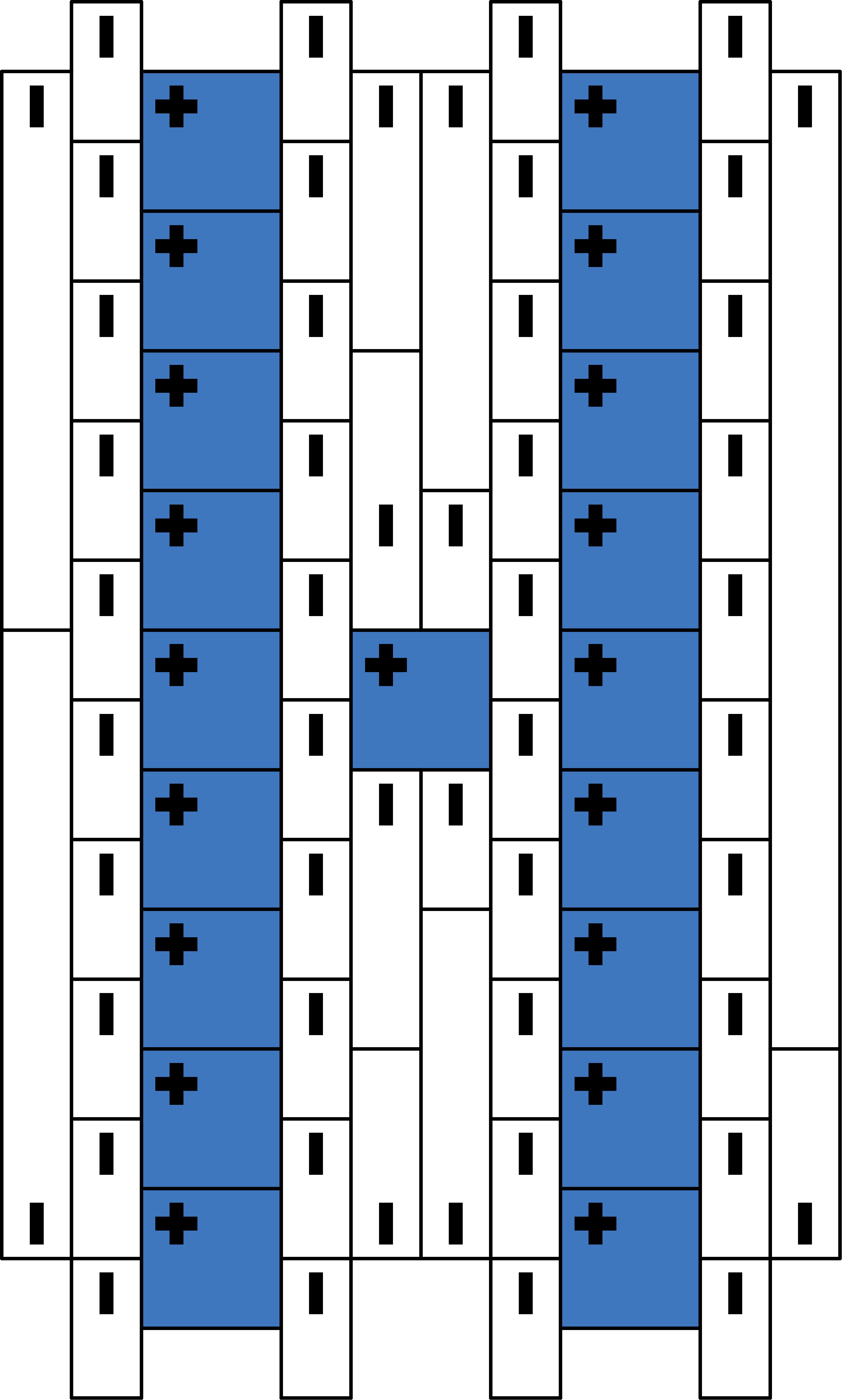}\hfil\hfil
\includegraphics[scale=\SCALE,trim=0 -75 0 0]{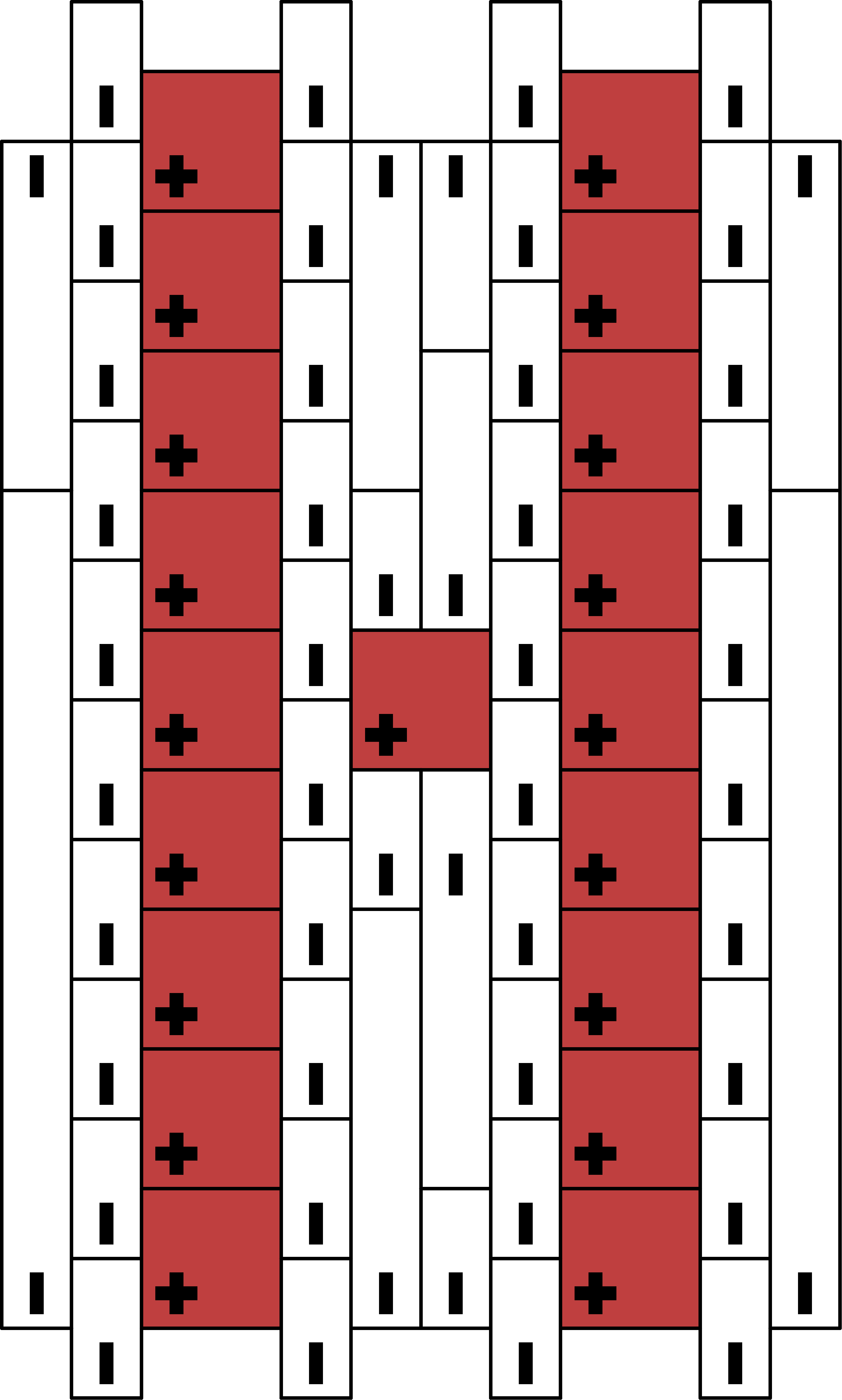}
\caption{The variable gadget's profile table. Left: variable set to true. Right: variable set to false.}\label{fig:variable}
\end{figure}

\begin{figure}
\centering
\includegraphics[scale=\SCALE]{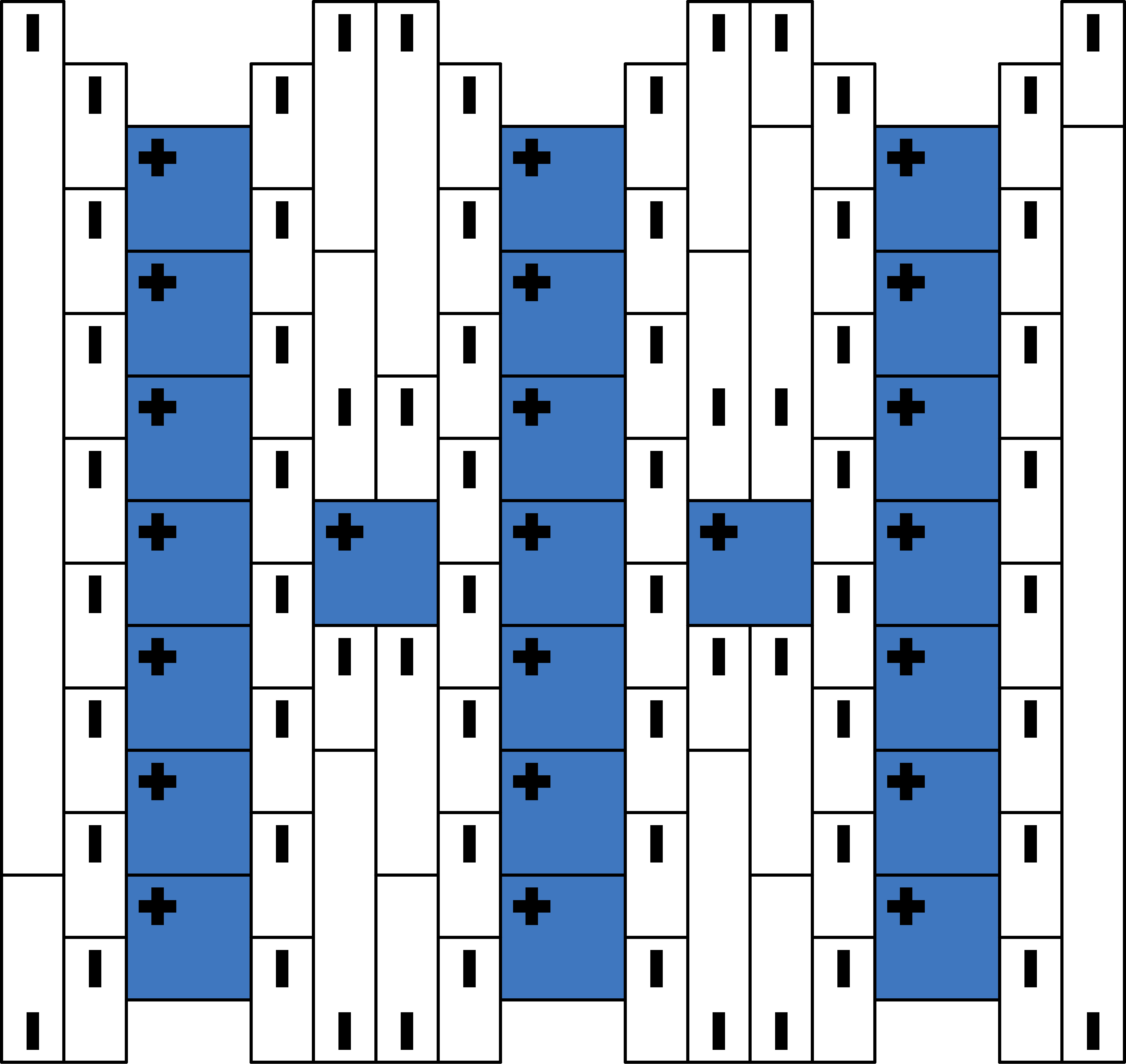}
\caption{A variable gadget widened to provide three wires, shown here set to true.}
\label{fig:variable-wide}
\end{figure}

\subsection{Clause Gadgets} \label{sec:clauses}

The clause gadget, shown in \cref{fig:clause-unsolved}, interfaces with three wire gadgets representing the three literals of this clause.  In the upper-left of the variable gadget is an internal wire, which we call the \defn{clause verification wire}.  The only way to cover the top two cells of that wire is using the wire's top \plusclue.  This is only possible when at least one of the wires is true, allowing a \defn{variable enforcement line} (drawn in figures as a purple horizontal bar) to provide a parity shift to the clause verification wire.  Otherwise, either those top two cells cannot be covered, or some other cell in the clause will not be covered, or there will be a four-corner violation.

\begin{figure}
\centering
\includegraphics[scale=\SCALE]{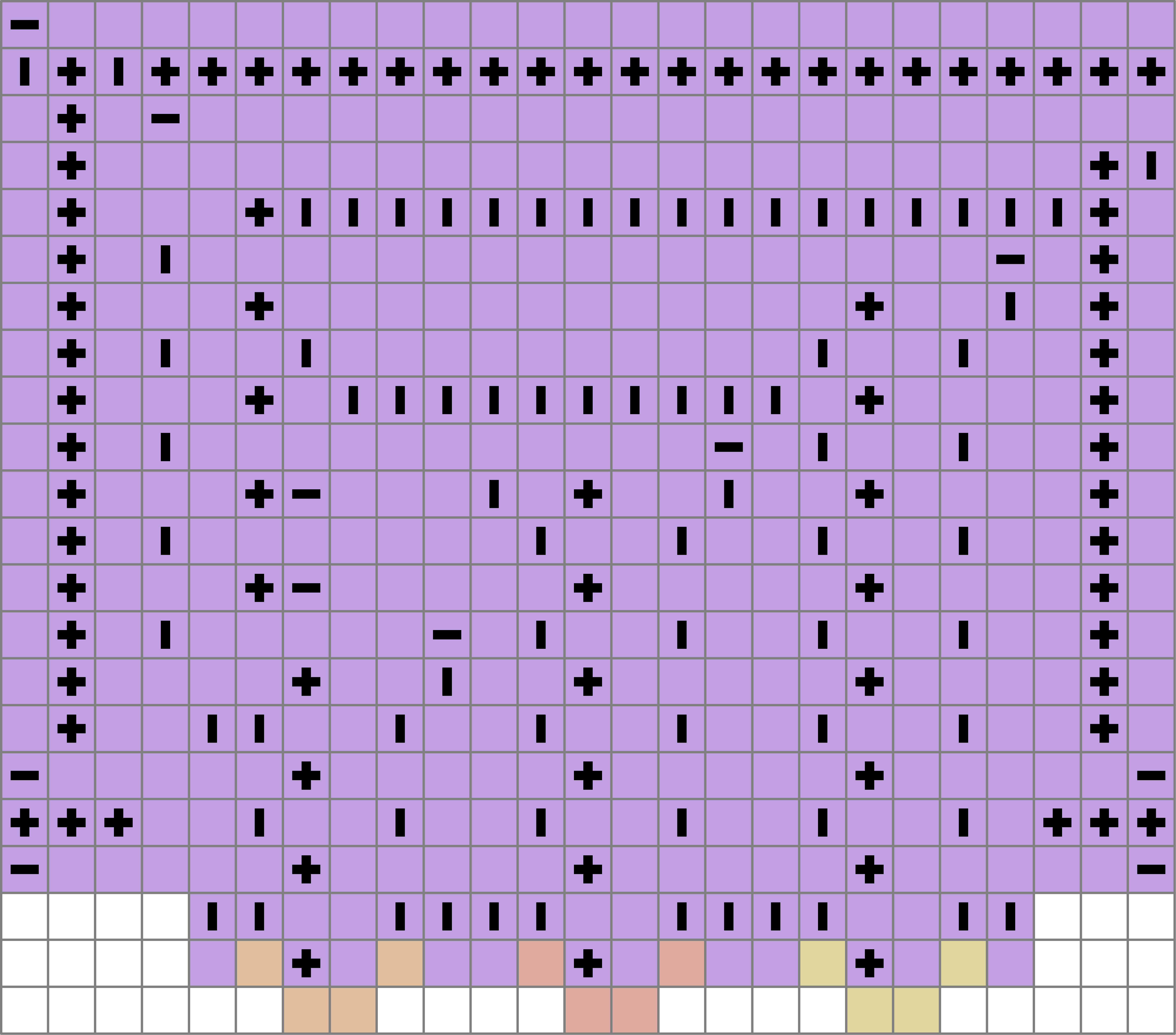}
\caption{An unsolved clause gadget.  Mandatory area is purple and optional areas are brown.}
\label{fig:clause-unsolved}
\end{figure}

The mandatory areas of the clause include all clues and cells shown in \cref{fig:clause}
and optional areas consisting of the row of cells at the bottom of the gadget,
specifically the set of cells under the \vertclue lines at the bottom of the gadget. \xxx{come back with colored figure}

Each of the three wires in this gadget has two intended solutions: true or false. In \cref{fig:clause},
the wire is blue if it represents true and red if it represents false.
The leftmost wire behaves somewhat differently from the others because
it is closest to the clause verification wire.

Importantly, the clause gadget can be expanded horizontally such that the variable
wires can be spaced an arbitrary amount beyond the width of the base gadget shown in
\cref{fig:clause}. The columns between the literal wires in the clause gadget can
 be expanded an arbitrary number of columns. Such an example expansion is shown in
 \cref{fig:wide-clause}. In this example, the columns have been expanded such that
 the entire gadget is wider by $4$ columns and the number of columns between
 each literal in the gadget has been expanded by $2$ columns.

 \begin{figure}[H]
    \centering
    \includegraphics[scale=\SCALE]{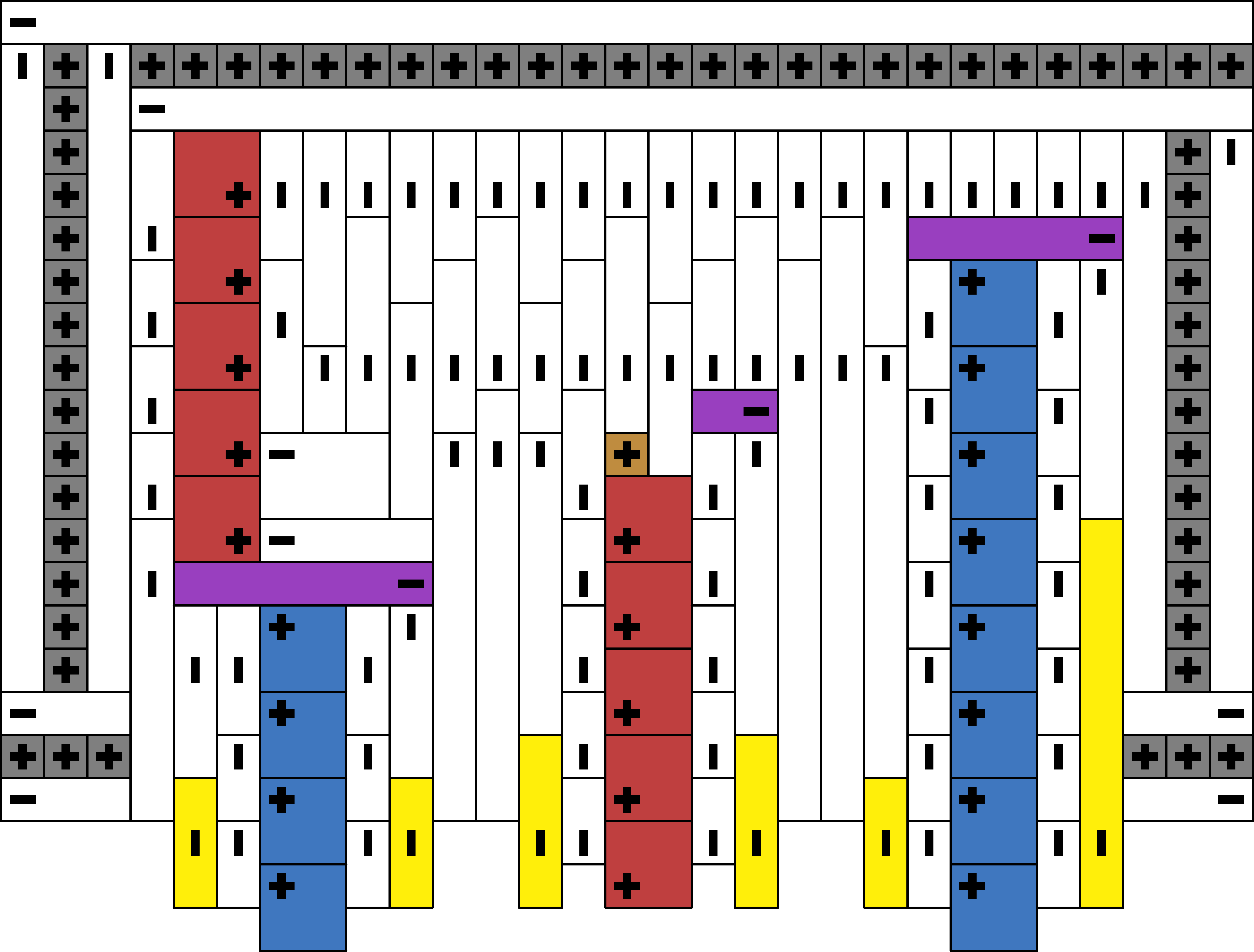}

    \vspace*{-2ex}

    {}\hspace*{\calc{10/29} \WIDTH}
    $\underbrace{\hspace*{\calc{2/29} \WIDTH}}_{\clap{repeatable}}$
    \hspace*{\calc{5.4/29}\WIDTH}
    $\underbrace{\hspace*{\calc{2/29} \WIDTH}}_{\clap{repeatable}}$
    \hspace*{\calc{9/29} \WIDTH}

	\caption{Example where the columns in between literal wires in the clause
    gadget have been expanded. The columns which are able to be repeated an arbitrarily
    number of times have been labeled as ``repeatable'' in the figure since
    they can be repeated an arbitrarily number of times to make the clause an
    arbitrary width.}\label{fig:wide-clause}
 \end{figure}

\begin{lemma}\label{lem:false}
  If any wire is in the \emph{false} configuration, then the variable enforcement line corresponding to the wire will not be able to go across the gadget.
\end{lemma}

\begin{proof}
If a wire is in the false configuration, then there exists at two cells on the top of the wire that need to be covered.
These two cells can be covered in two different ways. We first prove this lemma for the leftmost
wire and then prove the lemma for the other wires since the leftmost wire is different from the others.
In this case, the only way to cover the two cells is with a $2 \times 2$ square (see \cref{fig:leftmost}), blocking the variable enforcement line from crossing the top of the wire.

\begin{figure}[H]
\centering
\includegraphics[scale=\SCALE]{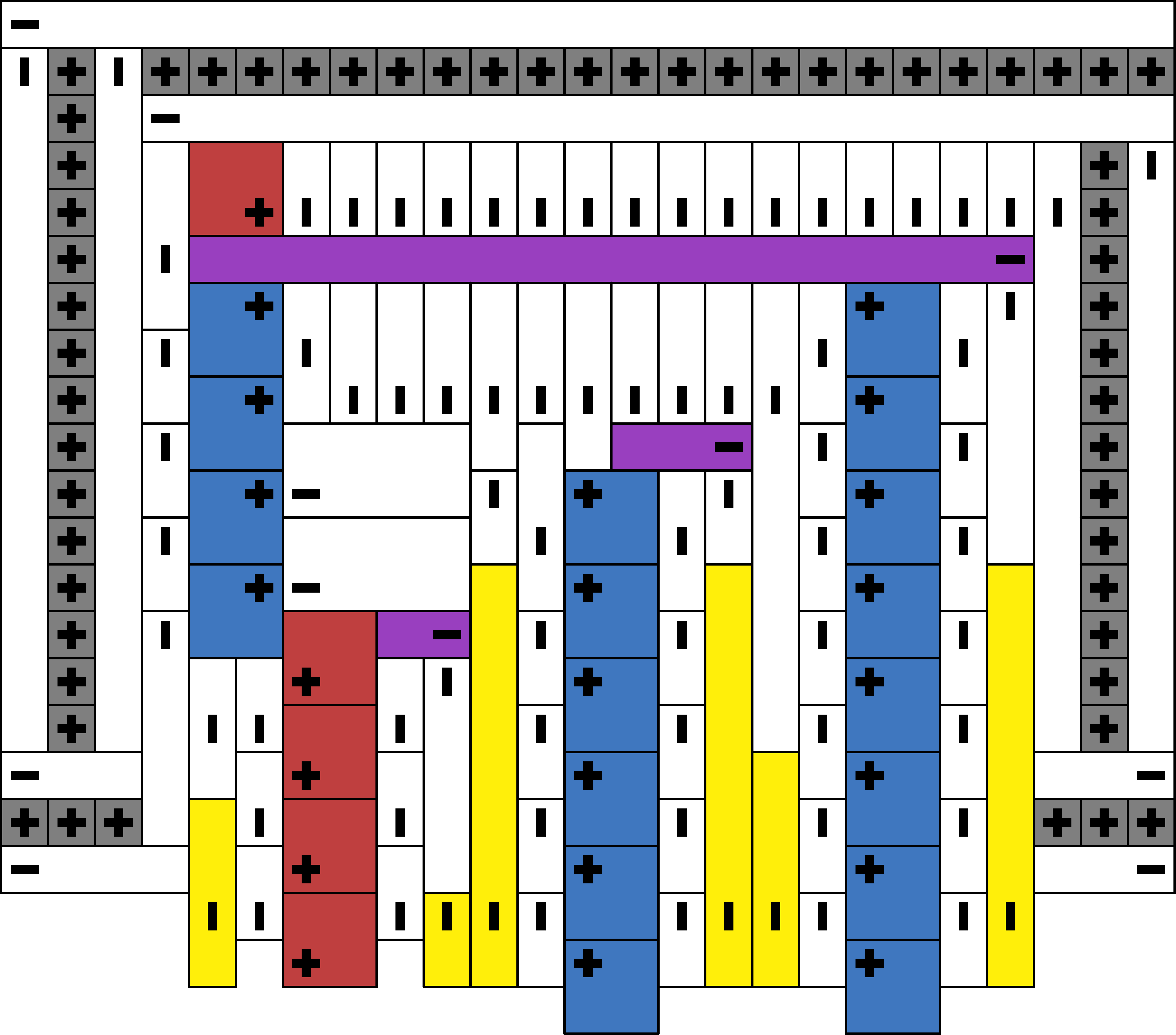}
\caption{When the leftmost wire is set in the false configuration, the only way to cover the top two cells
of the wire is with a $2 \times 2$ square that blocks the variable enforcement line.}\label{fig:leftmost}
\end{figure}

For the other two wires, the top two cells can be covered in only two ways.
Either a $1 \times 1$ square covers one of the two cells and a vertical line from the top
covers the other cell or vice versa (see \cref{fig:other-wires}).

\begin{figure}[H]
	\begin{subfigure}{0.48\linewidth}
	\centering
	\includegraphics[width=\linewidth]{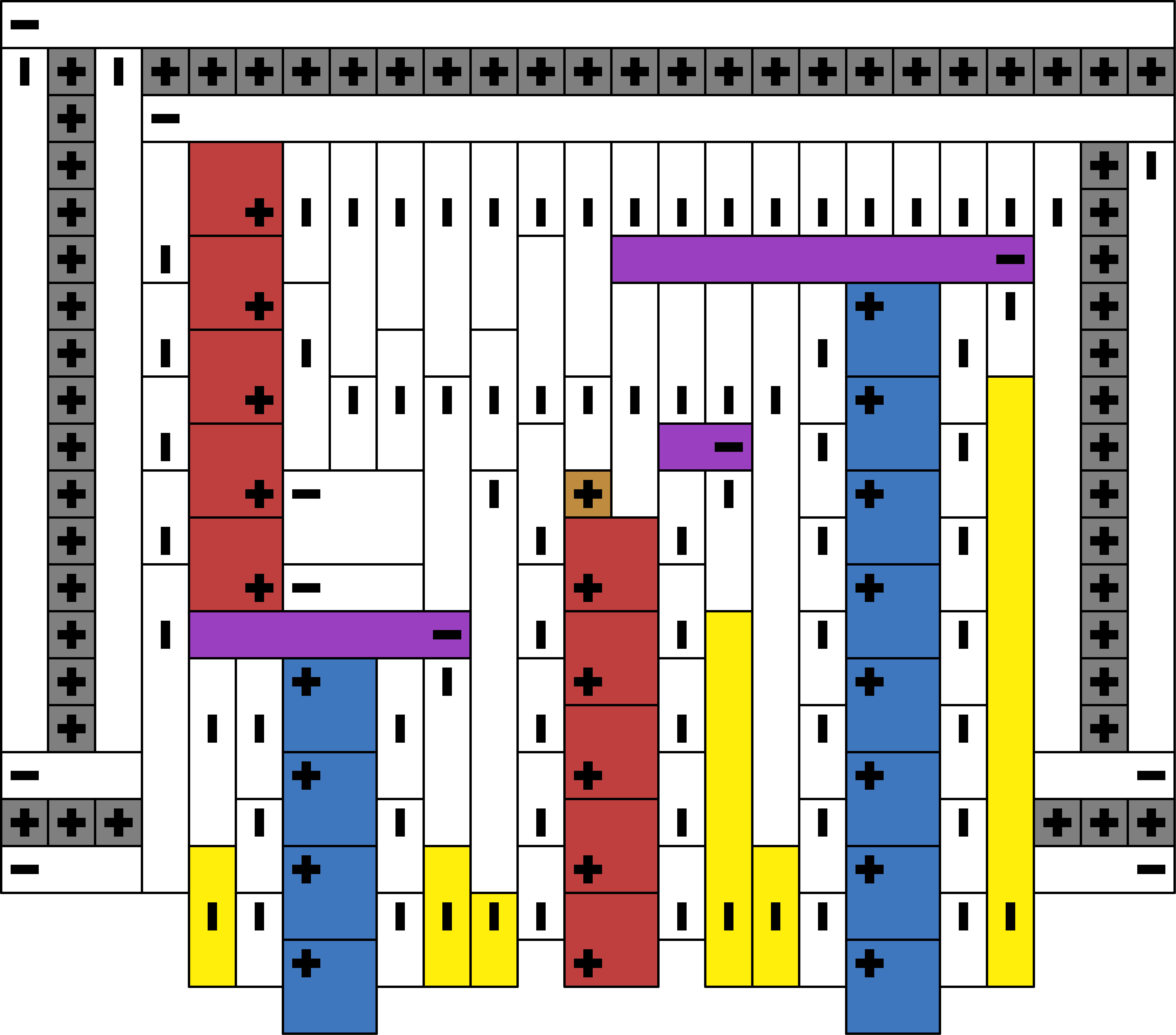}
	\caption{The middle wire is false.}\label{fig:wire-false-examples-1}
	\end{subfigure}\hfill
	\begin{subfigure}{0.48\linewidth}
	\centering
	\includegraphics[width=\linewidth]{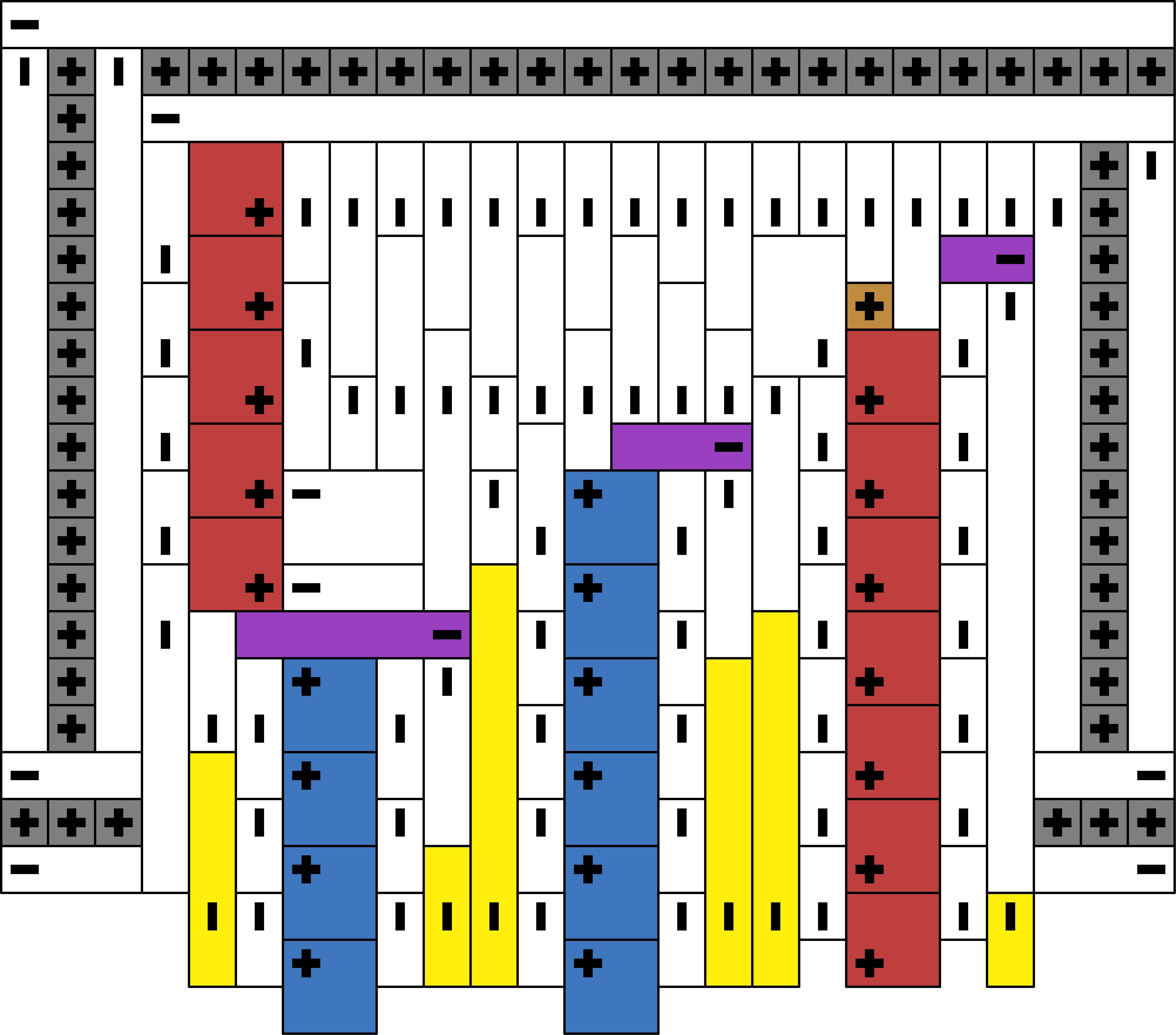}
	\caption{The right wire is false.}\label{fig:clause-false-examples-2}
	\end{subfigure}
	\caption{For the false wires, the only two configurations that will
	guarantee that the two cells at the top of the wires are covered are
	the cases where one $1 \times 1$ square covers one of the two cells and
	a long rectangle extending from the top covers the other cell. The $1\times 1$
	squares are shown in brown in the figures above.}\label{fig:other-wires}
\end{figure}

No other configurations are available that does not
violate the four-corner constraint. Thus, this configuration prevents the corresponding variable
enforcement lines from going across the gadget.
\end{proof}

\begin{corollary}\label{cor:all-false}
When all wires in the gadget are false, the puzzle does not have a solution.
\end{corollary}

\begin{proof}
By \cref{lem:false}, no variable enforcement line can go across the gadget if all
wires are false. In order to solve the puzzle presented by the gadget, the top two cells
of the clause verification line must be covered. These two cells cannot be covered by the
horizontal line on top of them nor can they be covered by the vertical lines beside them.
Thus, they must be covered by the $2 \times 2$ square formed in the clause verification line.
However such a square will either leave a cell in the middle of the clause verification line
uncovered or will leave the bottom two cells of the line uncovered. In this case, no configurations
exist in covering these bottom two cells without violating the four-corner rule. See \cref{fig:clause-all-false}.
Thus, the gadget is unsatisfiable if all wires into the gadget are false.
\end{proof}

\begin{lemma} \label{lem:one-true}
If at least one of the wires entering the clause gadget is in the true
configuration, then the clause gadget is locally solvable.
\end{lemma}

\begin{proof}
In any wire is in the true configuration, then the variable enforcement line corresponding to the gadget will be able to go across the
gadget. For the leftmost wire, the clause verification line will be in the configuration that ensures that all cells that need to be covered by the
line are covered. Otherwise, the variable enforcement line will be able to cause the clause verification line to cover all the necessary cells. See \cref{fig:clause-2,fig:clause-3,fig:clause-4,fig:clause-5,fig:clause-6,fig:clause-7,fig:clause-8}.
\end{proof}

Using the above lemmas, we are able to prove the following properties
of the profile table of the clause gadget.

\begin{corollary}\label{cor:table-proper}
    The profile table of the clause gadget is proper.
\end{corollary}

\begin{lemma}\label{cor:table-complete}
    The profile table of the clause gadget is complete.
\end{lemma}

\begin{proof}
The clause gadget's profile table contains all profiles shown in \cref{fig:clause} except for the all-false configuration shown in \cref{fig:clause-all-false}.  By \cref{cor:all-false}, the all-false configuration is not locally solvable.  It remains to show the all-false configuration is locally impossible.

To do this, we show that no solution to a clue outside of this profile
is able to solve any part of the all-false clause profile--essentially
that the clause gadget is fully isolated from the rest of the puzzle.
By design, no clue above, to the left of, or to the right of the clause can cover any of the cells that are left uncovered by the literals, because the row and columns of single-cell squares blocks any rectangles from reaching the uncovered cells.

We now prove that no clues from the bottom of the
gadget can help cover any of these cells. Such clues can only potentially
cover the optional areas at the bottom of the gadget. We show that
such clues cannot cover parts of the literal gadgets. By \cref{lem:coupler-configs},
there are only two possible configurations of the variable gadgets; thus,
no other outside fillers can cover any cells in the incoming wires. Hence,
no clues adjacent to the bottom of the gadget can help cover any part of
the incoming wires.

Thus the all-false profile is locally impossible, so the profile table is complete.
\end{proof}

\begin{figure}
\centering
\begin{subfigure}{0.45\linewidth}
\centering
\includegraphics[width=0.7\linewidth]{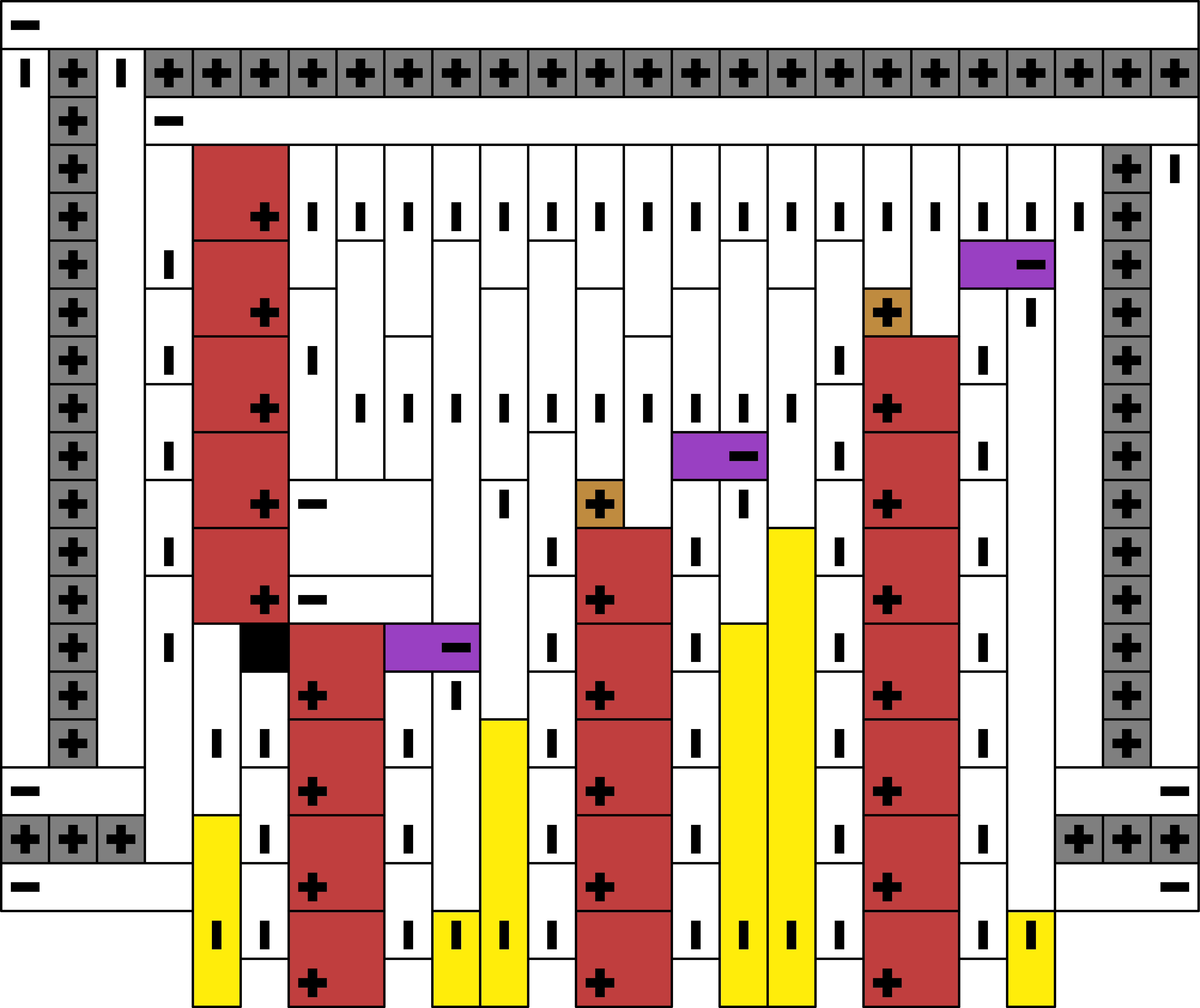}
\caption{The (false, false, false) configuration.}\label{fig:clause-all-false}
\end{subfigure}
\begin{subfigure}{0.45\linewidth}
\centering
\includegraphics[width=0.7\linewidth]{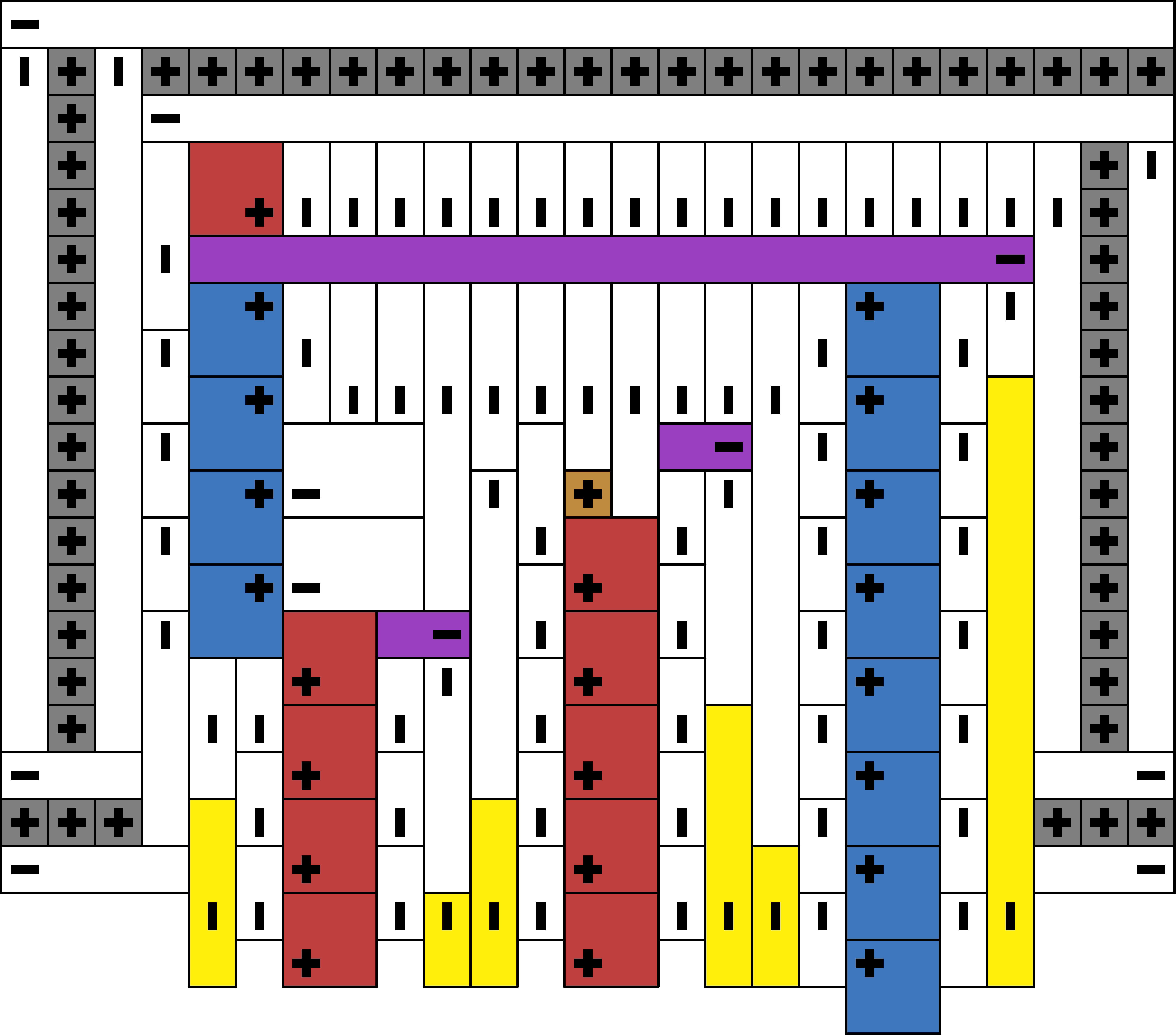}
\caption{The (false, false, true) configuration.}\label{fig:clause-2}
\end{subfigure}

\begin{subfigure}{0.45\linewidth}
\centering
\includegraphics[width=0.7\linewidth]{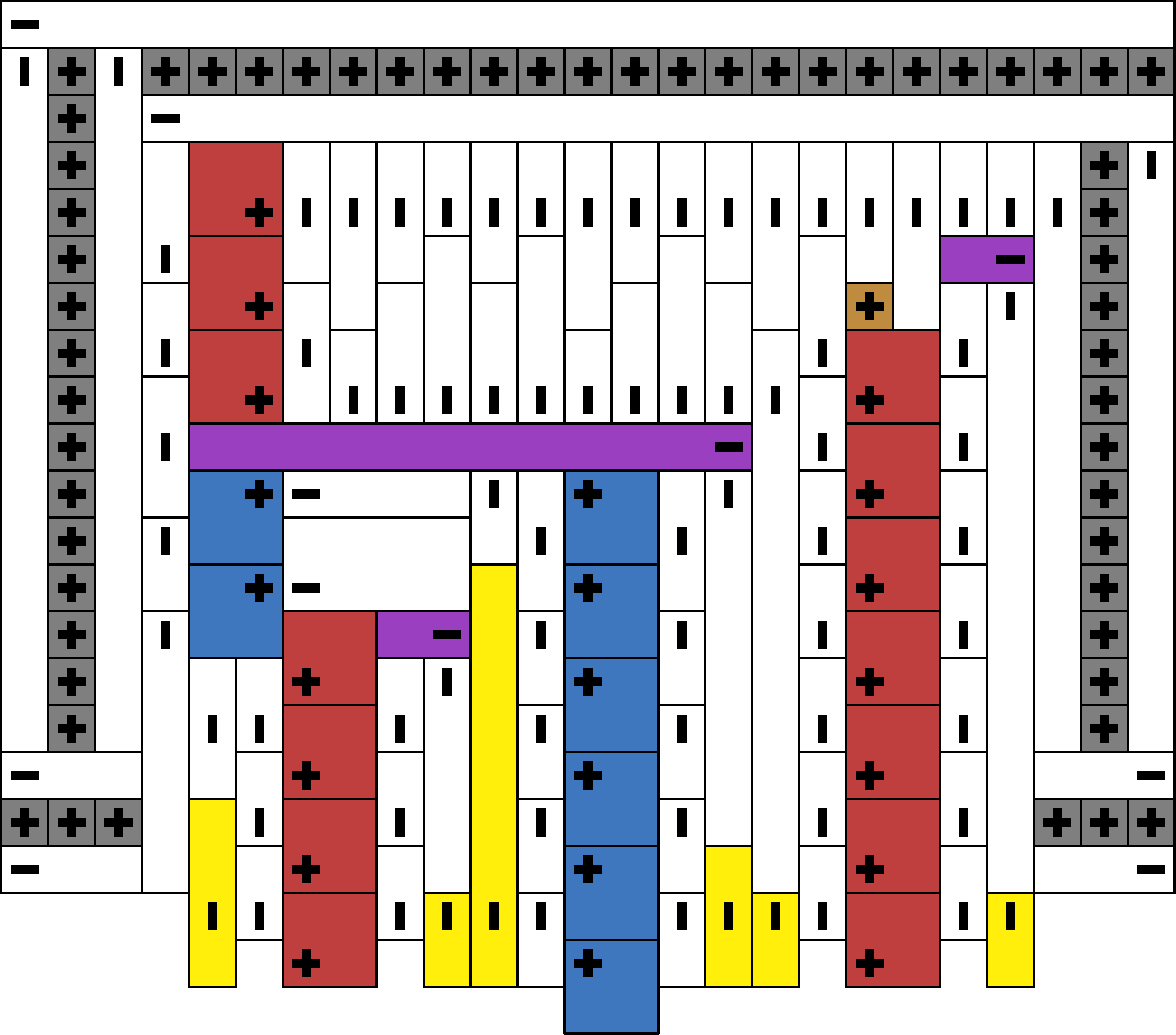}
\caption{The (false, true, false) configuration.}\label{fig:clause-3}
\end{subfigure}
\begin{subfigure}{0.45\linewidth}
\centering
\includegraphics[width=0.7\linewidth]{figures/clause-011}
\caption{The (false, true, true) configuration.}\label{fig:clause-4}
\end{subfigure}

\begin{subfigure}{0.45\linewidth}
\centering
\includegraphics[width=0.7\linewidth]{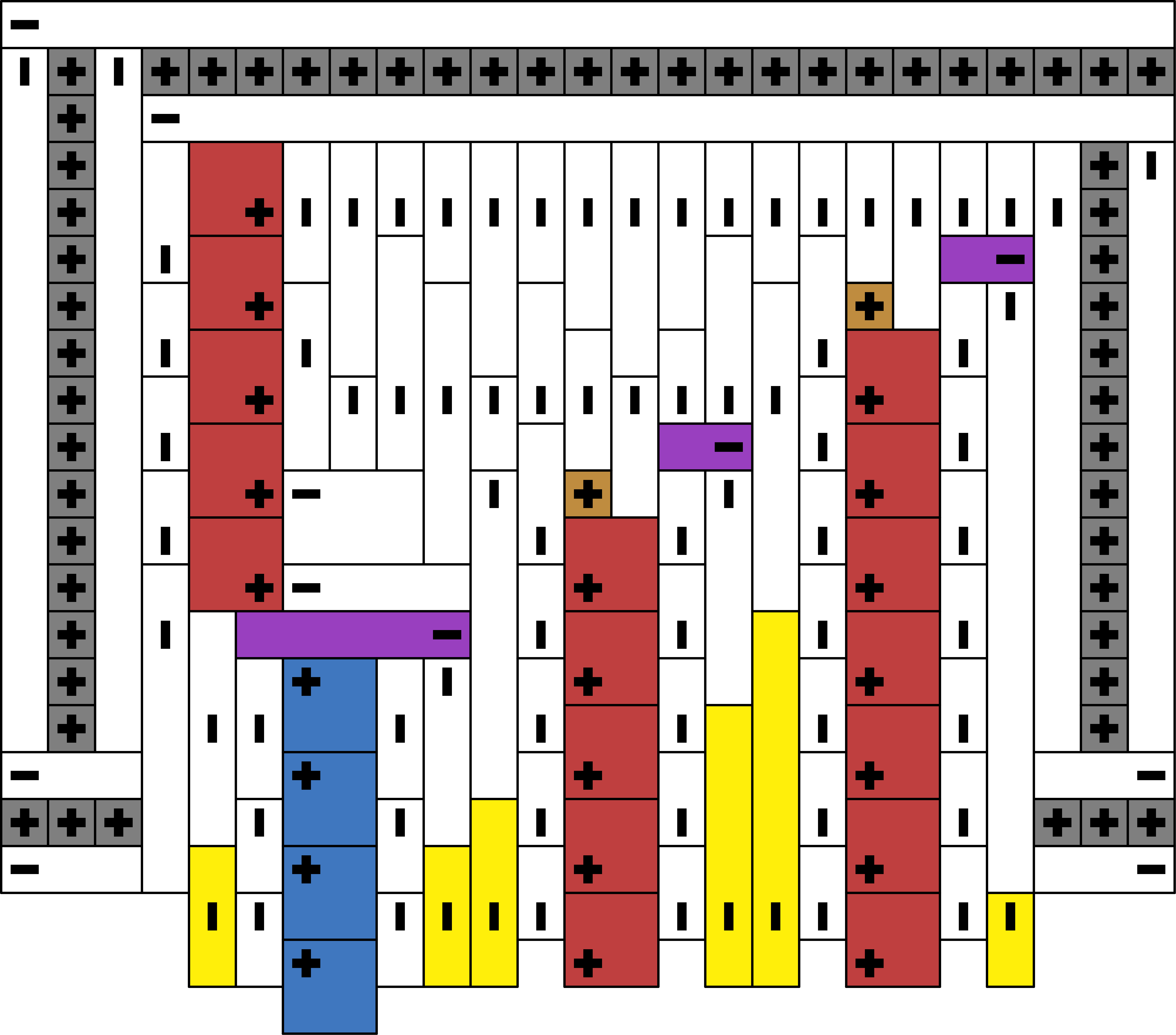}
\caption{The (true, false, false) configuration.}\label{fig:clause-5}
\end{subfigure}
\begin{subfigure}{0.45\linewidth}
\centering
\includegraphics[width=0.7\linewidth]{figures/clause-101}
\caption{The (true, false, true) configuration.}\label{fig:clause-6}
\end{subfigure}

\begin{subfigure}{0.45\linewidth}
\centering
\includegraphics[width=0.7\linewidth]{figures/clause-110}
\caption{The (true, true, false) configuration.}\label{fig:clause-7}
\end{subfigure}
\begin{subfigure}{0.45\linewidth}
\centering
\includegraphics[width=0.7\linewidth]{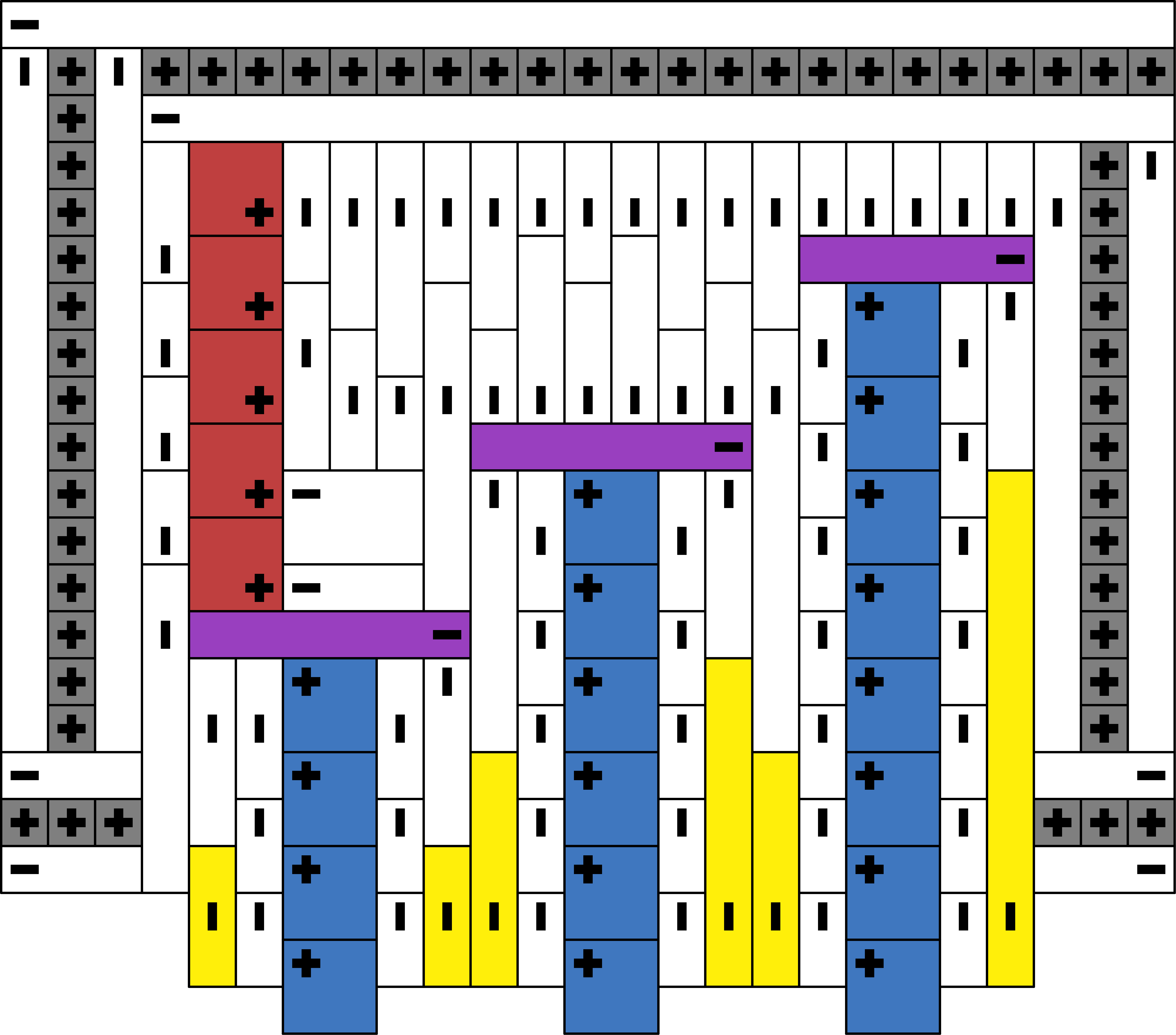}
\caption{The (true, true, true) configuration.}\label{fig:clause-8}
\end{subfigure}
\caption{The clause gadget.  All configurations shown here except the all-false configuration in \cref{fig:clause-all-false} are in the clause gadget profile table. Clues highlighted in yellow also function as the ``outer sheathing'' protecting the wires closest to them (see \cref{sec:filler})}\label{fig:clause}
\end{figure}

\subsection{Layout, Sheathing, and Filler} \label{sec:filler}

In order to build the full Tatamibari instance corresponding to a planar rectilinear monotone 3SAT instance, we lay out the gadgets as shown in \cref{fig:filler-diagram}: variable gadgets are positioned on a central line, while positive and negative clauses are positioned above and below respectively at heights corresponding with how many layers of clauses are nested below them, with wires running vertically from variables to clauses (both variable and clause gadgets can be extended arbitrarily far horizontally).  Variable and clause gadgets have rectangular profiles (except for where the wires ``plug in'' to them).  Variables and clauses have a uniform height, and for any two variable or clause gadgets, they are placed on exactly the same set of rows or they share no rows.

All wire gadgets in the puzzle produced by our reduction are safely placed; that is, no rectangle from a \vertclue can reach the cells to the right of the top and bottom \plusclues in the wire.  The only \vertclues in those columns are in clause gadgets.  The row of single-cell squares at the top of the clause gadget blocks any rectangles from extending upwards out of the clause gadget.  If a rectangle from a \vertclue in those columns of the clause gadget extends downward past the first \plusclue in the column to its left, the cell below that \plusclue cannot be covered by any clue, so rectangles cannot extend downward out of the clause gadget in those columns.  Thus \vertclues from clause gadgets cannot interact with wire gadgets, so the wires are safely placed.

Because we want the solvability of the Tatamibari instance to depend only on solving the gadgets, we need to add \defn{filler} clues that are always able to cover the areas outside the gadgets.

First, we set aside any cells horizontally adjacent to a wire gadget.  These cells will be covered by the outer sheathing clues described in described in \cref{sec:wires} and \cref{sec:variables} and highlighted in yellow in \cref{fig:wide-clause} and \cref{fig:clause}.  In the global solution, the areas assigned to the outer sheathing clues thus extend vertically outside their gadget.  For the purposes of the filler algorithm, we consider the cells covered by the outer sheathing to be part of the wire gadget.

Each filler clue corresponds to a rectangular area of space between gadgets, formed by breaking each row into maximal horizontally contiguous strips between (and bordered by) the gadgets, then joining vertically contiguous strips into a single rectangle if they have the same width.  The filler algorithm places a single clue in each of these rectangles (\CLUE|, \CLUE-, or \CLUE+ depending on the rectangle's aspect ratio), placed arbitrarily inside (say, in the upper-right corner).  See \cref{fig:filler-diagram} for an example. While it may be possible for the solver to use these clues differently than shown here, we only need to prove that if the solver does assign each rectangular area to its associated clue, it will cover the area.

The only potential problem lies in the possibility of a four-corner violation involving a filler rectangle. This can only happen where either (i) a corner of a filler rectangle meets a gadget and a wire coming from that gadget, or (ii) where two corners of filler rectangles meet along the edge of a gadget. If a corner of a filler section meets an edge of another filler section or the edge of the board there cannot be a four-corner violation.

\paragraph{Remark:} There is a potential third problem case, where two wires are directly adjacent with only the outer sheathing ($2$ columns) between them (see \cref{fig:clause-unsolved}, which has this property). This can be dealt with in either of two ways: ensuring that no wires are directly adjacent to each other by stretching the instance horizontally, or noting that the meeting points of the outer sheathing of the two adjacent wires can be adjusted to not produce a four-corner violation between them.

\begin{proposition} \label{prop: filler}
If the gadgets can all be satisfied, the filler clues can also be satisfied.
\end{proposition}

\begin{proof}
Each filler clue will be satisfied by a rectangle covering its entire associated area; the cells horizontally adjacent to wires will be filled by two width-$1$ vertical rectangles from the outer sheathing clues, one coming from the clause gadget above and the other coming from the variable gadget below. The meeting point between the two outer sheathing rectangles can be adjusted as needed to avoid a four-corner violation. As mentioned, we have only two problem cases: (i) a corner of the filler rectangle meets a gadget and protruding wire; and (ii) corners of two sections meet on the side of a wire. Because both cases involve the side of a wire, we can avoid violations in either case by appropriately adjusting the meeting point of the sheathing clues.

(i)  To avoid having a corner where the corner of the filler section meets the wire and gadget, the meeting point of the two sheathing clues can be placed on the edge (not corner) of the filler section, thus avoiding a four-corner violation since the corner of the filler section meets the edge of one of the sheathing rectangles.

(ii) As long as the meeting point of the two sheathing rectangles of the wire is not at the point where the two filler sections meet, there is no four-corner violation. The meeting point can trivially be placed on the side of a filler section (while still respecting the parity of the wire as expressed by the inner sheathing).

Therefore, since the sheathing can always be adjusted to accommodate filler rectangles, the satisfiability of the Tatamibari instance depends only on the gadgets.
\end{proof}

\subsection{Finale} \label{sec:finale}

Now we can show that Tatamibari is NP-hard. Let $f$ be the reduction, which takes an instance $\Phi$ of planar rectilinear monotone 3SAT and returns a Tatamibari instance $f(\Phi)$; we want to show:
\begin{proposition}
If $\Phi$ has $n$ variables and $m$ clauses, then $f(\Phi)$
has size polynomial in $n+m$, and can be computed in time polynomial in $n+m$.
\end{proposition}

\begin{proof}
Our construction expands the given planar rectilinear monotone 3SAT instance
by a constant factor.
Therefore it suffices to prove that planar rectilinear monotone 3SAT
is strongly NP-hard when given the coordinates of the rectilinear drawing.
Indeed, the height of the drawing is $O(m)$ and the width of the drawing
is $O(e)$ if the graph has $e$ edges, which is $O(m+n)$ by planarity.
\end{proof}

\begin{proposition}
If $\Phi$ has a solution, then $f(\Phi)$ also has a solution.
\end{proposition}

\begin{proof}
We begin by taking the solution to $\Phi$ and setting the variable gadgets' profiles according to those values; by \cref{lem: variable valid}, they will all be locally solvable. By \cref{lem:one-true}, since each clause gadget is connected to wires representing variables which satisfy the clause, there must be a solution to the clause gadget. Furthermore, by \cref{prop: filler}, if the gadgets are satisfiable then the rest of the space can be filled without contradiction, producing a solution to $f(\Phi)$.
\end{proof}

\begin{proposition}
If $\Phi$ has no solution, then $f(\Phi)$ also has no solution.
\end{proposition}

\begin{proof}
We prove the equivalent statement that if $f(\Phi)$ has a solution, then $\Phi$ must also have a solution.

First, we prove that any solution to $f(\Phi)$ must correspond to some setting of the variables $x_1, \dots, x_n$ of $\Phi$. This is a consequence of \cref{lem: variable gadget}, which shows that all wires in a single variable gadget must carry the same value, which is then taken as the setting for that variable.

Next, we have to show that these settings of the variables $x_i$ are a solution of $\Phi$; to do this, note that by \cref{cor:wire-parity} each wire ending in a clause gadget must carry its value into this clause gadget; and by \cref{cor:all-false} and \cref{lem:one-true} there is a solution to the clause gadget if and only if the wires represent values which satisfy the clause.

Thus, the values of the variable gadgets must be a solution to $\Phi$.
\end{proof}

The above three propositions imply our desired result: %
\begin{theorem}
Tatamibari is (strongly) NP-hard.
\end{theorem}
Because a given Tatamibari solution can be trivially checked in polynomial time, this theorem implies that Tatamibari is NP-complete.

\section{Font}
\label{sec:font}

\cref{puzzle font} shows a series of twenty-six $10 \times 10$
Tatamibari puzzles that we designed,
whose unique solutions shown in \cref{solved font}
reveal each letter A--Z.
To represent a bitmap image in the solution of a Tatamibari puzzle,
we introduce two colors for clues, light and dark, and similarly shade
the regions corresponding to each clue.
As shown in \cref{solved font},
the letter is drawn by the dark regions from dark clues.
These puzzles were designed by hand, using our SAT-based solver \cite{github}
to iterate until we obtained unique solutions.
The font is also available online.%
\footnote{\url{http://erikdemaine.org/fonts/tatamibari/}}

\def\scale{1.2}

\begin{figure}
  \centering
  \includegraphics[scale=\scale]{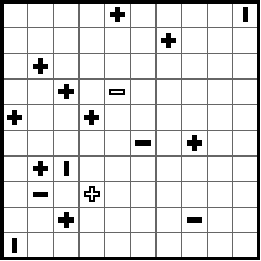}
  \includegraphics[scale=\scale]{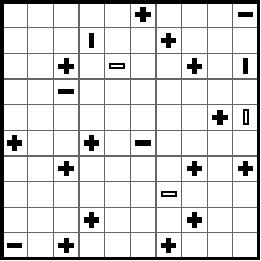}
  \includegraphics[scale=\scale]{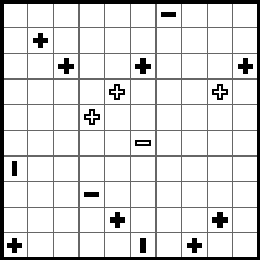}
  \includegraphics[scale=\scale]{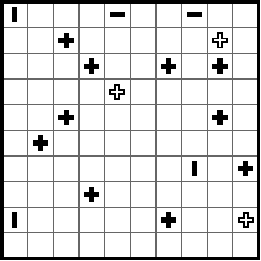}
  \includegraphics[scale=\scale]{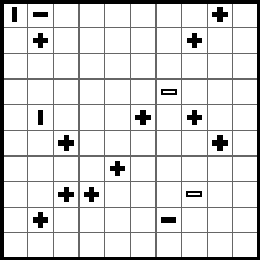}

  \smallskip
  \includegraphics[scale=\scale]{font/puzzle-F}
  \includegraphics[scale=\scale]{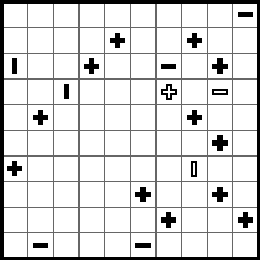}
  \includegraphics[scale=\scale]{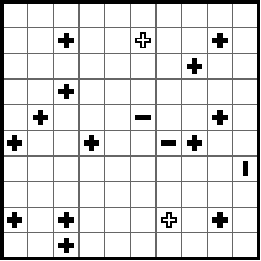}
  \includegraphics[scale=\scale]{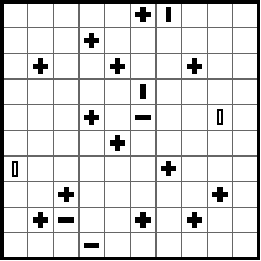}
  \includegraphics[scale=\scale]{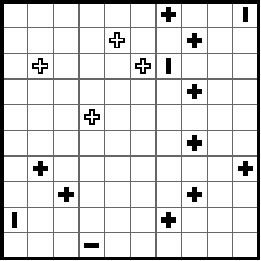}

  \smallskip
  \includegraphics[scale=\scale]{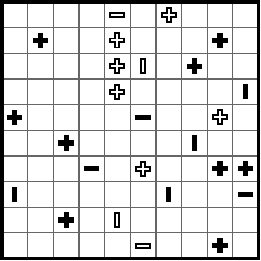}
  \includegraphics[scale=\scale]{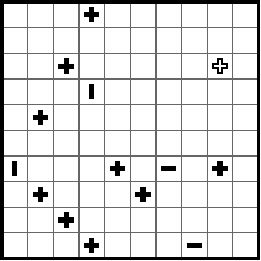}
  \includegraphics[scale=\scale]{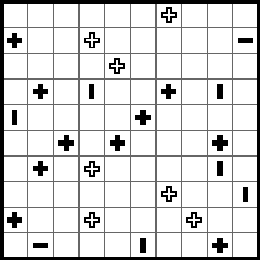}
  \includegraphics[scale=\scale]{font/puzzle-N}
  \includegraphics[scale=\scale]{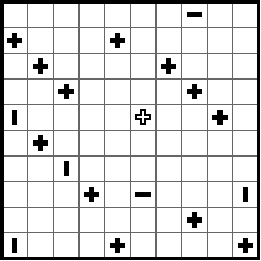}

  \smallskip
  \includegraphics[scale=\scale]{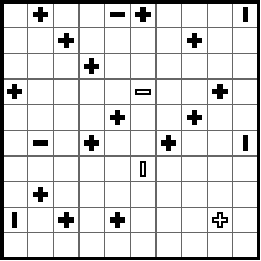}
  \includegraphics[scale=\scale]{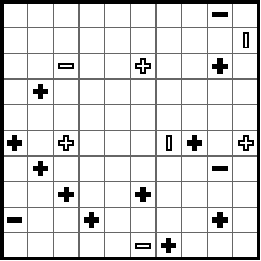}
  \includegraphics[scale=\scale]{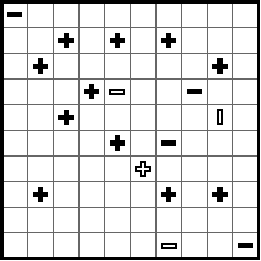}
  \includegraphics[scale=\scale]{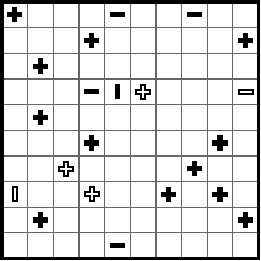}
  \includegraphics[scale=\scale]{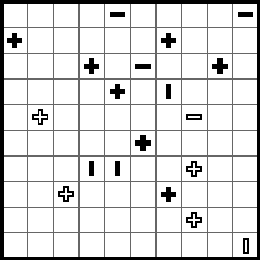}

  \smallskip
  \includegraphics[scale=\scale]{font/puzzle-U}
  \includegraphics[scale=\scale]{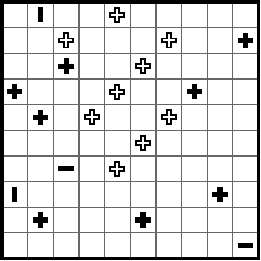}
  \includegraphics[scale=\scale]{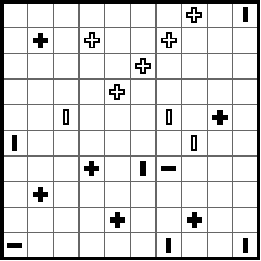}

  \smallskip
  \includegraphics[scale=\scale]{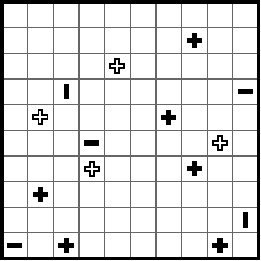}
  \includegraphics[scale=\scale]{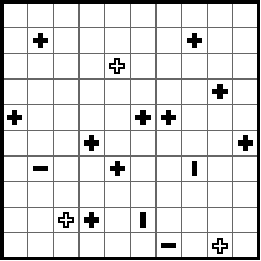}
  \includegraphics[scale=\scale]{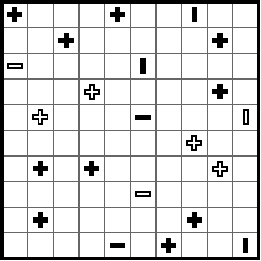}

  \caption{Puzzle font: each puzzle has a unique solution whose regions for
    dark clues (shown in \protect\cref{solved font})
    form the shape of a letter.}
  \label{puzzle font}
\end{figure}

\begin{figure}
  \centering
  \includegraphics[scale=\scale]{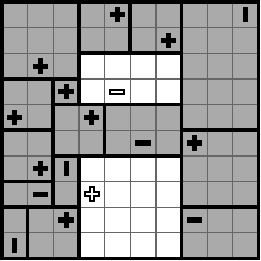}
  \includegraphics[scale=\scale]{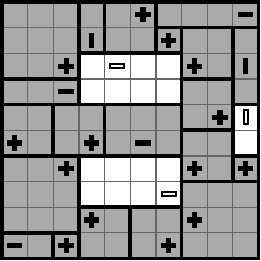}
  \includegraphics[scale=\scale]{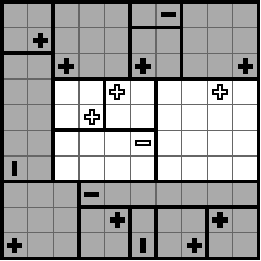}
  \includegraphics[scale=\scale]{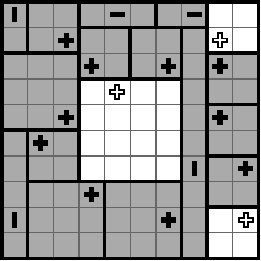}
  \includegraphics[scale=\scale]{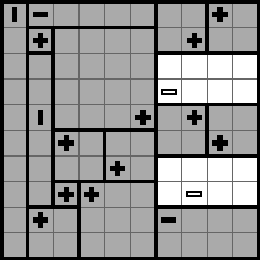}

  \smallskip
  \includegraphics[scale=\scale]{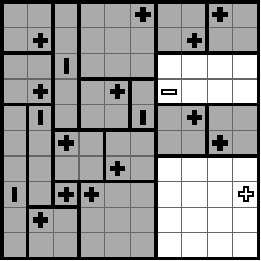}
  \includegraphics[scale=\scale]{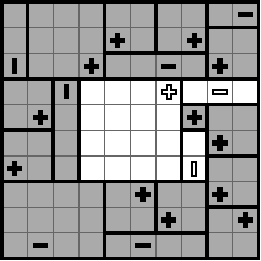}
  \includegraphics[scale=\scale]{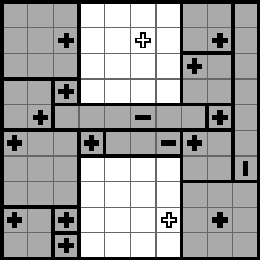}
  \includegraphics[scale=\scale]{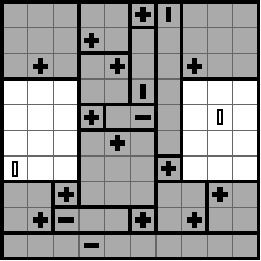}
  \includegraphics[scale=\scale]{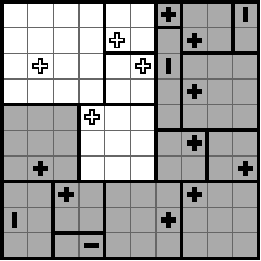}

  \smallskip
  \includegraphics[scale=\scale]{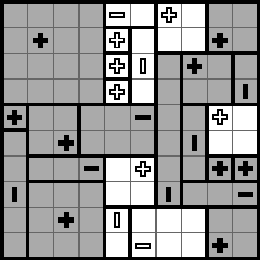}
  \includegraphics[scale=\scale]{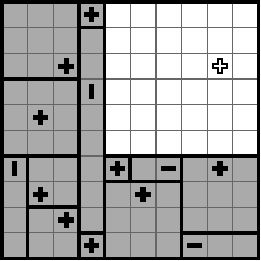}
  \includegraphics[scale=\scale]{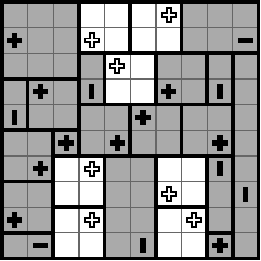}
  \includegraphics[scale=\scale]{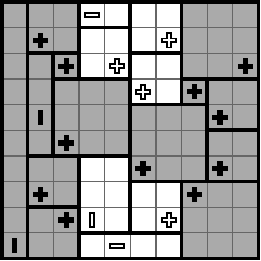}
  \includegraphics[scale=\scale]{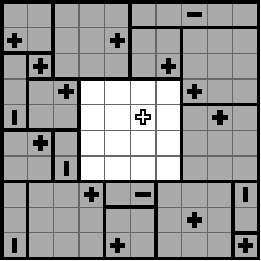}

  \smallskip
  \includegraphics[scale=\scale]{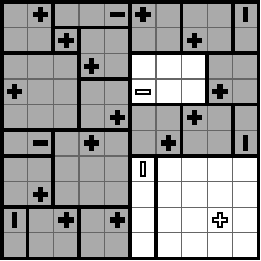}
  \includegraphics[scale=\scale]{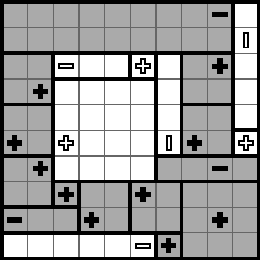}
  \includegraphics[scale=\scale]{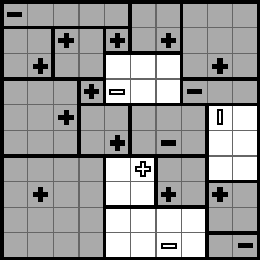}
  \includegraphics[scale=\scale]{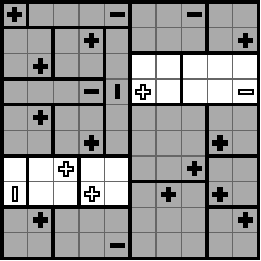}
  \includegraphics[scale=\scale]{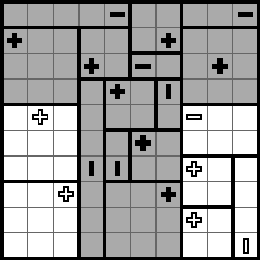}

  \smallskip
  \includegraphics[scale=\scale]{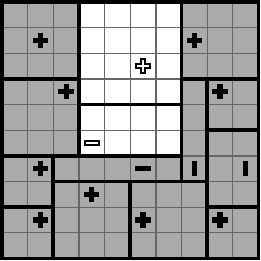}
  \includegraphics[scale=\scale]{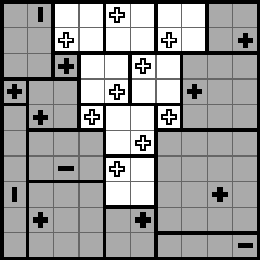}
  \includegraphics[scale=\scale]{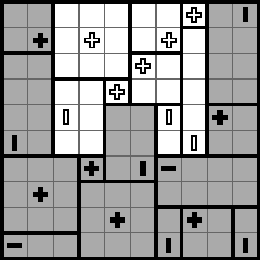}

  \smallskip
  \includegraphics[scale=\scale]{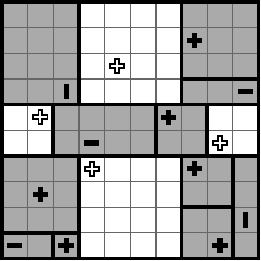}
  \includegraphics[scale=\scale]{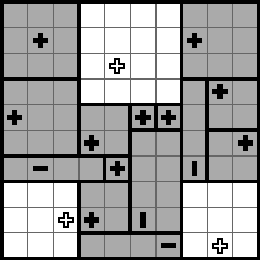}
  \includegraphics[scale=\scale]{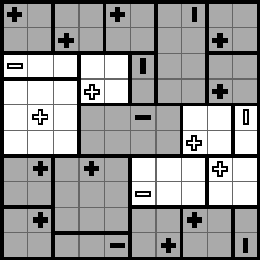}

  \caption{Solved font: unique solutions to the puzzles in \protect\cref{puzzle font}.}
  \label{solved font}
\end{figure}

\begin{figure}
  \centering
  \def\scale{1.4}
  \includegraphics[scale=\FUNscale]{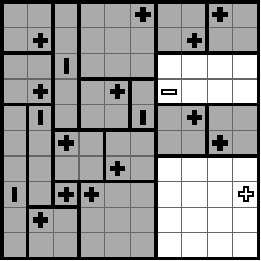}\qquad
  \includegraphics[scale=\FUNscale]{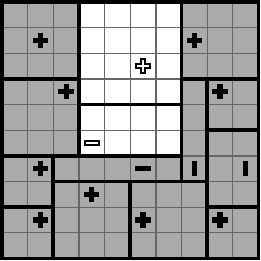}\qquad
  \includegraphics[scale=\FUNscale]{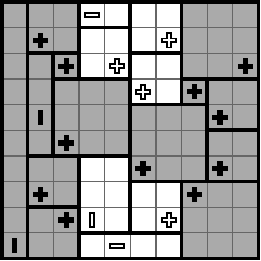}
  \caption{Solution to \cref{fig:FUN}.}
  \label{fig:FUN solved}
\end{figure}

\section{Open Problems} \label{sec:open-problems}
There remain interesting open questions regarding the computational complexity of Tatamibari.
When designing puzzles, it is often desired to have a single unique solution.
We suspect that Tatamibari is ASP-hard (NP-hard to determine whether it has another solution, given a solution), and that counting the number of solutions is \#P-hard.
However, our reduction is far from parsimonious.
Some rework of the gadgets, and a unique filler between gadgets, would be required to preserve the number of solutions.

We could ask about restrictions or natural variations of Tatamibari.
For example, we are curious whether a Tatamibari puzzle with only \plusclues, or only \vertclues, remains hard.
We have also wondered about the version of the puzzle without the four-corner constraint.
Although initially we thought of the four-corner constraint as a nuisance to be overcome in our reduction, our final proof uses it extensively and centrally.

\appendix
\section{Example: Spiral Galaxies} \label{sec:spiral-galaxies}

As an example of the gadget area hardness framework, we show how the NP-hardness proof for Spiral Galaxies from~\cite{SpiralGalaxies} can be described using the framework.  A Spiral Galaxies puzzle is a rectangular grid with clues in some grid cells or on some grid lines.  The goal is to partition the puzzle into areas with a single clue per area such that the area is rotationally symmetric about the clue.

We reduce from planar\footnote{Friedman's proof \cite{SpiralGalaxies} provides a crossover gadget, but it is not necessary because AND and NOT build a crossover~\cite{10.1145/1008354.1008356}.} Boolean circuit satisfiability.  We have a wire gadget, a variable gadget, NOT and AND gadgets, a fanout (wire duplicator) gadget, and a vertical shift gadget.  We lay out these gadgets to overlap in their optional areas (only), and communicate a truth value in whether the optional area is covered or not.

\newcommand\spiralscale{25}
\begin{figure}
  \centering
  \subcaptionbox{Unsolved wire gadget}{\includegraphics[scale=\spiralscale]{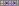}}\quad
  \subcaptionbox{Wire carrying true signal: $3 \times 2$ rectangles}{\includegraphics[scale=\spiralscale]{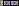}}\quad
  \subcaptionbox{Wire carrying false signal: alternating $1 \times 2$ and $5 \times 2$ rectangles}{\includegraphics[scale=\spiralscale]{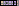}}
  \caption{Spiral Galaxies wire and its profile table (true and false solutions)}
  \label{fig:spiral-wire}
\end{figure}

\paragraph{Wire.}  The wire gadget consists of repeating pairs of clues three grid units apart.  There are two gadget solutions, shown in \cref{fig:spiral-wire}: repeating $3 \times 2$ rectangles, in which case the wire covers the right optional area, and alternating $1 \times 2$ and $5 \times 2$ rectangles, in which case the wire covers the left optional area.  The wire carries a true signal when it covers the right optional area and false when it covers the left optional area.  The wire gadget can be extended to arbitrary length in units of two clues.  (The proof in \cite{SpiralGalaxies} does not explicitly state this parity requirement, but the gate gadgets assume the true signal protrudes into the gadget to cover the optional input area and the false signal does not.)

Boolean circuit satisfiability requires the circuit produce a true output.  We can force a wire to be true simply by terminating it.  Because the wire has height two, any filler clues to the right of the wire cannot cover area in the wire gadget, so the wire must end in a $3\times 2$ rectangle to cover the right optional area, forcing the rest of the wire to also carry a true signal.

\begin{figure}
  \centering
  \subcaptionbox{Unsolved variable gadget}{\includegraphics[scale=\spiralscale]{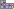}}\quad
  \subcaptionbox{Variable set to true}{\includegraphics[scale=\spiralscale]{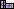}}\quad
  \subcaptionbox{Variable set to false}{\includegraphics[scale=\spiralscale]{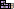}}
  \caption{Spiral Galaxies variable and its profile table (true and false solutions)}
  \label{fig:spiral-variable}
\end{figure}

\paragraph{Variable.}  The variable gadget is shown in \cref{fig:spiral-variable}.  There are two gadget solutions, one leaving the optional area uncovered (so the adjacent wire is set to true) and the other covering it (so the adjacent wire is set to false).  Choosing one solution or the other corresponds to assigning true or false to the variable.

\begin{figure}
  \centering
  \subcaptionbox{Unsolved NOT gadget}{\includegraphics[scale=\spiralscale]{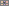}}\quad
  \subcaptionbox{NOT gadget converting true to false}{\includegraphics[scale=\spiralscale]{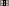}}\quad
  \subcaptionbox{NOT gadget converting false to true}{\includegraphics[scale=\spiralscale]{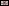}}
  \caption{Spiral Galaxies NOT gadget and its profile table}
  \label{fig:spiral-not}
\end{figure}

\paragraph{NOT.}  The NOT gadget is shown in \cref{fig:spiral-not}.  If the left optional area is covered by the input wire (carrying a true signal), the clue in the NOT gadget must cover a $1 \times 2$ rectangle, so the right optional area must be covered by the output wire carrying a false signal.  If the left optional area is uncovered (when the input wire is false), the clue in the NOT gadget covers both optional areas, so the output wire must carry a true signal.  Thus the NOT gadget inverts the wire's signal.

\begin{figure}
  \centering
  \subcaptionbox{Unsolved AND gadget (inputs at left)}{\includegraphics[scale=\spiralscale]{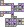}}

  \subcaptionbox{AND with true and true inputs}{\includegraphics[scale=\spiralscale]{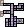}}\quad
  \subcaptionbox{AND with true and false inputs}{\includegraphics[scale=\spiralscale]{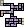}}

  \subcaptionbox{AND with false and true inputs}{\includegraphics[scale=\spiralscale]{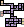}}\quad
  \subcaptionbox{AND with false and false inputs}{\includegraphics[scale=\spiralscale]{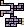}}
  \caption{Spiral Galaxies AND gadget and its profile table}
  \label{fig:spiral-and}
\end{figure}

\paragraph{AND.}  The AND gadget is shown in \cref{fig:spiral-and}.  When both inputs are true, both of the left optional areas are covered by the wire, so the clues to the right of the optional area are rectangles and the clue at the center of the gadget is a long vertical rectangle, allowing the right optional area to be covered, propagating a true signal from the gadget.  When either or both of the inputs are false, one or both of the left optional areas must be covered by the clue(s) to the right of the areas, blocking the central clue from covering a vertical rectangle, preventing the right optional area from being covered, thus propagating a false signal from the gadget.

\begin{figure}
  \centering
  \subcaptionbox{Unsolved fanout gadget (input at top)}{\includegraphics[scale=\spiralscale]{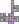}}

  \subcaptionbox{Fanout gadget duplicating true wire}{\includegraphics[scale=\spiralscale]{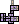}}\quad
  \subcaptionbox{Fanout gadget duplicating false wire}{\includegraphics[scale=\spiralscale]{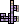}}
  \caption{Spiral Galaxies fanout gadget and its profile table}
  \label{fig:spiral-fanout}
\end{figure}

\paragraph{Fanout.}  Like the AND gadget, the fanout gadget (\cref{fig:spiral-fanout}) is also based around forming or not forming a vertical rectangle.  The upper optional area is the input.  When it is covered by the input wire (a true signal), the central clue cannot form a vertical rectangle, so the upper-right optional area must be covered by the clue to its left, and because the bottom cell in the central column is covered by the clue to its upper-left, the lower-right optional area must also be covered by the clue to its left.  When the upper optional area is not covered by the input wire, it must be covered by the central clue forming a vertical rectangle, so the output optional areas cannot be covered by the clues to their left.

\begin{figure}
  \centering
  \subcaptionbox{Unsolved shift gadget}{\includegraphics[scale=\spiralscale]{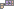}}\quad
  \subcaptionbox{Shift gadget shifting a true wire}{\includegraphics[scale=\spiralscale]{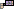}}\quad
  \subcaptionbox{Shift gadget shifting a false wire}{\includegraphics[scale=\spiralscale]{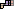}}
  \caption{Spiral Galaxies upward shift gadget and its profile table; the downward shift gadget is this gadget flipped vertically}
  \label{fig:spiral-shift}
\end{figure}

\paragraph{Shift.}  Because variable and gate outputs are on the right and gate inputs are on the left, we do not need a turn gadget, but we do need to shift wires vertically, which is done using the shift gadget.  An upward shift gadget is shown in \cref{fig:spiral-shift}; the downward shift gadget is that gadget's reflection across the horizontal axis.  When the input wire is true, the input wire covers the left optional area, so the left clue is covered by a single cell and the right clue covers the right optional area, propagating true on the output.  When the input wire is false, the left clue covers the left optional area and forces the right clue to be a $1 \times 2$ rectangle, leaving the right optional area uncovered, propagating false on the output.

\paragraph{Layout.}  Friedman's proof in \cite{SpiralGalaxies} omits discussion of layout, but we sketch a layout algorithm here.  We start with a grid embedding of the input planar Boolean circuit.  We scale the grid by at least 6 so that our wire gadget fits for unit-length wires, but possibly by a greater factor if the grid embedding has long vertical segments, because our shift gadget consumes horizontal distance to move vertically.

\paragraph{Filling algorithm.}  The filling algorithm places a clue in the center of every cell that isn't part of a gadget, forcing them to be covered by single-cell areas.  Filler clues could only cover area in a gadget if two cells in the gadget area are separated by one filler clue and those cells do not themselves have clues.  This is avoided in all gadgets by ensuring all gadget cells that are separated by filler are separated by two or more filler cells, so only local gadget solutions are possible.

\paragraph{Composition algorithm.}  The local gadget solutions are already consistent with each other, so to form an area assignment for the entire puzzle, the composition algorithm simply takes the local gadget solutions and assigns each filler clue to the cell containing it.

\paragraph{Proper and complete profile tables.}  The profile tables are proper because they contain only proper profiles.  Because the filler clues cannot cover area in the gadgets, we can verify by case analysis that the profile tables are complete (all other profiles are locally impossible).  This completes the proof.

\section*{Acknowledgments}

This work was initiated during open problem solving in the MIT class on
Algorithmic Lower Bounds: Fun with Hardness Proofs (6.890)
taught by Erik Demaine in Fall 2014.
We thank the other participants of that class
for related discussions and providing an inspiring atmosphere.

\bibliographystyle{alpha}
\bibliography{ref}
\end{document}